\documentclass[11pt]{article}
\usepackage{graphicx}
\usepackage{amsfonts}
\usepackage{amsmath}
\usepackage{amssymb}
\usepackage{url}
\usepackage{fancyhdr}
\usepackage{indentfirst}
\usepackage{enumerate}
\usepackage{dsfont,footmisc}
\usepackage{comment}

\usepackage{setspace}

\usepackage[misc]{ifsym}

\usepackage{amsthm}
\usepackage{color}
\usepackage{natbib}
\usepackage[british]{babel}
  \usepackage[colorlinks=true,citecolor=blue]{hyperref}
\usepackage[right,pagewise, mathlines]{lineno}

\usepackage{tikz}

\usetikzlibrary{decorations.pathreplacing}
\usepackage{tikz-qtree}

\usetikzlibrary{arrows.meta,positioning}

\newcommand{\var}{\mathrm{Var}}
\newcommand{\cov}{\mathrm{Cov}}

\newcommand{\dd}{\mathrm{d}}

\newcommand{\E}{\mathbb{E}}
\newcommand{\R}{\mathbb{R}}

\newcommand{\N}{\mathbb{N}}
\newcommand{\p}{\mathbb{P}}

\newcommand{\id}{\mathds{1}}

\newcommand{\X}{\mathcal X}

\newcommand{\XX}{\mathbf X}

\renewcommand{\ge}{\geqslant}
\renewcommand{\le}{\leqslant}
\renewcommand{\geq}{\geqslant}
\renewcommand{\leq}{\leqslant}
\renewcommand{\epsilon}{\varepsilon}

\theoremstyle{plain}
\newtheorem{theorem}{Theorem}
\newtheorem{lemma}{Lemma}
\newtheorem{proposition}{Proposition}

\theoremstyle{definition}
\newtheorem{example}{Example}

\theoremstyle{remark}
\newtheorem{remark}{Remark}

\theoremstyle{definition}

\renewcommand{\cite}{\citet}
\renewcommand{\cdots}{\dots}

\usepackage{algorithm}
\usepackage{algpseudocode}
\usepackage{tikz}
\usepackage{enumerate}

\usepackage{dsfont}
\usepackage{newtxmath} 

\usepackage{tikz}

\usetikzlibrary{decorations.pathreplacing}
\usepackage{tikz-qtree}

\usetikzlibrary{arrows.meta,positioning}

\renewcommand{\ge}{\geqslant}
\renewcommand{\le}{\leqslant}
\renewcommand{\geq}{\geqslant}
\renewcommand{\leq}{\leqslant}
\renewcommand{\epsilon}{\varepsilon}

\newenvironment{axiom}[2]{%
    \vspace{8pt} 
    \par\noindent\textbf{Axiom #1 (#2).}\ }{%
    \vspace{8pt} 
    \par}

\newenvironment{propty}[2]{%
    \vspace{8pt} 
    \par\noindent\textbf{Property #1 (#2).}\ }{%
    \vspace{8pt} 
    \par}

\begin{document} 
	

\title{Anonymized Risk Sharing}

\author{Zhanyi Jiao\thanks{School of Mathematics and Statistical Sciences, Arizona State University, USA. \Letter~\url{zjiao123@asu.edu}} 
\and
Steven Kou\thanks{Department of Finance, Questrom School of Business, Boston University, USA. \Letter~\url{kou@bu.edu}}  
	\and
	Yang Liu\thanks{School of Science and Engineering, The Chinese University of Hong Kong, Shenzhen, China. \Letter~\url{yangliu16@cuhk.edu.cn}} 
	\and
	Ruodu Wang\thanks{Department of Statistics and Actuarial Science, University of Waterloo, Canada. \Letter~\url{wang@uwaterloo.ca}} 
}
\date{\today}
\maketitle

\begin{abstract}
Anonymized risk sharing requires no information about agents' preferences, identities, private operations, or realized losses. It is especially relevant in the digital economy, with applications such as P2P health-care insurance, revenue sharing for digital music and videos, and blockchain mining pools. Although there is an extensive literature on axiomatic approaches in decision theory, there is, so far, no axiomatic theory of risk sharing. We develop such a theory in the context of anonymized risk sharing. Equilibrium and various notions of optimality are discussed. Applications to the digital economy are presented.

\end{abstract}

\section{Introduction}\label{sec:Intro}
Risk sharing 
refers to pooling risks from several participants in a group and reallocating the total risk.
A risk sharing scheme arises in many different business practices, such as insurance, tontines, profit sharing contracts in investment,
to name a few. In these contexts, either profits
or losses, or both of them, may be shared among participants.
The traditional approach to studying risk sharing is 
via equilibrium, either Pareto   or competitive equilibria,  dating back to the classic work of \cite{AD54} and \cite{B62};\footnote{See the later work on risk sharing by \cite{R79} for expected utility, 
\cite{BE05} for convex risk measures, \cite{CDG12} for multivariate stochastic dominance,  \cite{BHYZ15} and \cite{XZ16} for rank-dependent utilities, \cite{CLW17} for reinsurance arrangements, and \cite{ELW18} for quantile-based risk measures.} see, e.g., \cite{S11} for a general treatment.
In either form of equilibria, information on the preferences of the agents is required to define and compute an equilibrium.
Commonly used preferences include
 expected utility, rank-dependent utility, cumulative prospect theory, risk measures, and many more advanced models;  see \cite{G09} and \cite{W10} for decision models and \cite{FS16} for risk measures. 
However, in many practical situations, one rarely has precise information on the preferences, since elicitation of preferences can be challenging and costly (e.g., \cite{L83}), and preferences may be incomplete, ambiguous, or falsely supplied (e.g., \cite{DL18}). This is especially relevant in the digital economy, where anonymized digital platforms are used to implement risk sharing, such as blockchain mining pools, revenue sharing of digital music and movies, and P2P insurance.





In this paper, we consider a framework of \emph{anonymized risk sharing}  through a risk sharing platform,  such as the ones mentioned above. 
The key feature of 
this framework is that agents do not need to disclose their preferences, identities, or wealth levels.\footnote{We chose the term ``anonymized risk sharing" over ``anonymous risk sharing", as the former emphasizes that individual information is deliberately masked (but it could be available), and the latter stresses that such information is not known or supplied.} 
More precisely, the allocation to an agent is determined by the initial risk contributions of all agents, but not the specification of these agents.
For this reason, anonymized risk sharing schemes are desirable in several applications such as P2P insurance (e.g.,  \cite{D19} and  \cite{AF22}), Bitcoin mining pools (e.g., \cite{ES18}, \cite{LS20}  and \cite{CHL21}), tontines (e.g., \cite{CHK19} and \cite{HL22}), and revenue sharing of digital content (e.g., \cite{MKAC23} and \cite{BM25}).
Anonymized risk sharing is also closely related to the concept of decentralized risk sharing, where some key features are that no capital reserve is needed and agents settle their allocations via a pre-specified contract and a payment network; see  \cite{Feng2023}.  Examples of decentralized insurance include decentralized P2P insurance, mutual aid, and catastrophic risk pooling; see   \cite{FLZ22} for a summary of models for decentralized insurance.


Since no information on individual preferences is known in the setting of anonymized risk sharing, it is difficult to conduct equilibrium analysis; see \cite{MP89} for the discussion of information auditing in a principal-agent relationship (e.g., insurer-policyholder). Instead,  we propose an axiomatic approach  from the normative standpoint of the risk sharing platform. To our knowledge, this paper is the first axiomatic study on risk sharing rules.\footnote{Traditional axiomatic approaches (e.g., \citet{Y87}, \citet{GS89} and \citet{S89}) in  decision theory typically focus on mappings from  inputs to compare (called acts) to scalar numbers. In contrast, risk sharing rules are mappings from   inputs (random vectors in our context) to other  random vectors,  involving a more complicated structure.  This is perhaps the main reason that the axiomatic approach for risk  sharing is missing in the existing literature. 
}
We introduce four axioms, namely, actuarial fairness, risk fairness, risk anonymity, and operational anonymity. 
In particular, we show that the four axioms \emph{uniquely} 
characterize one risk sharing rule (Theorem \ref{th:full}),  the \emph{conditional mean risk sharing (CMRS)}. 
As an important risk sharing rule in economic theory with many attractive properties, CMRS was used by \cite{LM94} to study Pareto optimality of comonotonic risk allocations, and its  properties were studied in detail by \cite{DD12}; see \cite{DDR22} for a summary of these properties. 
In addition, we show in Theorem \ref{cor:UI} that, assuming risk anonymity and operational anonymity, CMRS is the unique risk sharing rule that is plausible for all risk-averse agents.   Our characterization hence provides a first axiomatic foundation for CMRS and its applications in finance and insurance.\footnote{In contrast, the axiomatic approach for (both deterministic and stochastic) preference functionals and risk measures (i.e., mappings from random variables to real numbers) has a rich literature. See,  e.g., 
\cite{MMR06}, and \cite{GS17} for preference functionals, and \cite{ADEH99} and \cite{FS02} for risk measures, as well as the more recent \cite{WZ21}.
}

On the technical side, the proof of Theorem \ref{th:full} relies on a new characterization of the conditional expectation which we present in Theorem \ref{thm:BT}. 
We further show that the four axioms are independent (Proposition \ref{prop:indep}). 
Several other properties related to our axioms are studied in Section \ref{sec:property}, including backtracking, universal improvement, and symmetry. 
Sections \ref{sec:equ} and \ref{sec:equilibrium} discuss Pareto optimality and competitive equilibrium with the preferences of agents and the connection to comonotonicity. CMRS is the unique rule that satisfies three economic principles as in Theorem \ref{prop:IR}. 
Section \ref{sec:generalized} further generalizes the above framework to include modeling of target information. The formulation of target information is motivated by the consideration that participants may sometimes prefer to settle the risk transfer using information beyond the total loss or profit incurred.  This formulation enables a broader range of practical applications.
The applications discussed in Section \ref{sec:app} consist of blockchain mining pools, digital music and video revenue sharing, and P2P insurance (including a case study of Xiang Hu Bao, a P2P mutual health care platform with over 100 million members, and the tontine scheme as a way to hedge longevity risk),  highlighting the suitability of our axiomatization and its implications for risk sharing mechanisms. Proofs of all results are relegated to the appendices.

\section{Risk Sharing Rules: Motivation and Notation}
\label{sec:2}
\vspace{1.5ex}
 
\subsection{Notation and Definitions}
\label{sec:21}

Let  $(\Omega,\mathcal F, \p)$ be a probability space and $\X$ be a set of random variables on this space, representing the set of possible random losses of interest. We assume that $\X$ is closed under addition and $0\in \X$.
Positive values of random variables represent losses, and negative values represent gains; flipping this convention makes no difference to all mathematical results. 
We always treat almost surely (a.s.) 
equal random variables as identical and  we omit ``a.s." in most places.
We use $\sup X$ for the essential supremum of $X$, i.e., $\sup X =\inf \{ x\in \R:\p(X>x)=0\}$.

In our framework, $n$ economic agents share a total risk, where $n\ge 3$ is an integer.
We assume $n\ge 3$ since the case $n=2$ is technically different and practically restrictive; see Example \ref{ex:1}
in Section \ref{sec:UI}.
For notational simplicity we write 
$[n]=\{1,\dots,n\}.$ 
Each agent $i\in [n]$ faces an initial risk $X_i\in \X$, a monetary outcome, which is the risk contribution of agent $i$ to the risk sharing pool. We use the term ``risk" to reflect that the random variable $X_i$ may be positive or negative, and sometimes we use the term ``loss" to emphasize its positive side.
Our framework and axioms can be equivalently formulated, with interpretations unchanged, by letting the positive side of $X_i$   represent gains instead of losses, which is convenient in some applications in Section \ref{sec:app}. 
For any random variable $S$, the set of all \emph{allocations} of $S$ is denoted by
$$
\mathbb{A}_n(S)=\left\{(Y_1,\ldots,Y_n)\in \mathcal{X}^n:~ \sum_{i=1}^nY_i=S\right\} . 
$$ 
We write $\mathbf X=(X_1,\dots,X_n)$ for the initial risk (contribution) vector, and $S^\XX =\sum_{i=1}^n X_i$ for the total risk. 

A \emph{risk sharing rule} is a mapping $\mathbf A:\X^n \to\X^n$ satisfying $\mathbf{A}^{\XX}=(A_1^{\XX},\dots,A_n^{\XX})\in \mathbb A_n(S^\XX)$ for each $\mathbf X \in \X^n$. Here, we use $\mathbf A^\XX$ instead of the usual notation $\mathbf A(\XX)$,  as it is easier to write components of $\mathbf A^\XX$.
The requirement $\mathbf{A}^{\XX}\in \mathbb A_n(S^\XX)$ means that  $\mathbf{A}^{\XX}$ sums up to the total risk, that is, full allocation. In other words, no external funds enter or leave the risk sharing pool except for the initial contributions of the agents, and there is no capital reserve needed for the risk exchange to settle. Full allocation is a requirement in decentralized risk sharing mechanisms.\footnote{It also allows for  centralized risk sharing, by including the central party as one of the participants.}

Each component of $\mathbf{A}^\XX$ represents the (random) allocation of risk to an agent.   
Through the rule $\mathbf A$, the initial risk vector $\XX$ enters the sharing pool as an input, and the allocation vector $\mathbf{A}^\XX$ comes out as the output.  Given each  scenario $\omega \in \Omega$, the actual payment is settled as the vector $\mathbf{A}^\XX(\omega) \in \R^n$. A positive payment  ${A}_i^\XX(\omega)=x>0$ means that agent $i$ needs to pay the amount of $x$, because positive values represent losses.
This   simple procedure is illustrated in Figure \ref{fig:risk}.

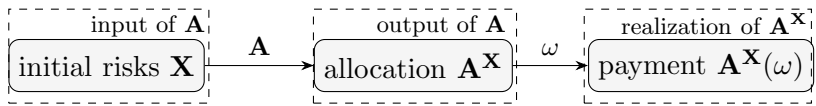
\begin{figure}[htp]
	\begin{center}
		\tikzstyle{bag} = [text width=5em, text centered]
		\tikzstyle{end} = []
		\begin{tikzpicture}[
			block/.style ={rectangle, draw=black, fill=gray!08,   align=center, rounded corners, minimum height=1.8em},
			myarrow/.style={-Stealth},
			node distance=1.5cm and 3.5cm
			] 
			\node[ block ] (c1) at (0,0)  {initial risks $\mathbf X$};  
			\node[ block ,   right= 1.45cm of c1] (c2) {allocation $\mathbf A^\XX$};     
			\draw [myarrow] (c1) -- node[sloped,font=\small, above] {$\mathbf A$}  
			(c2); 
			\draw [draw=black, dashed] (-1.3,-0.5) rectangle (1.36,0.75);
			\draw [draw=black, dashed] (2.75,-0.5) rectangle (5.45,0.75); 
			\node at (0.55,0.52) (c3) {\footnotesize input of $\mathbf A$};
			\node at (4.45,0.52) (c4) {\footnotesize output of $\mathbf A$};
			\node[ block ,   right= 1.0cm of c2] (c4) {payment $\mathbf A^\XX(\omega)$};     
			\draw [myarrow] (c2) -- node[sloped,font=\small, above] {$\omega$}  
			(c4); 
			\draw [draw=black, dashed] (6.35,-0.5) rectangle (9.37,0.75);  
			\node at (8.08,0.56) (c4) {\footnotesize realization of $\mathbf A^\XX$}; 
		\end{tikzpicture} 
	\end{center}
	\caption{Risk sharing}
	\label{fig:risk}
\end{figure}

As a key feature of this framework, different from the large body of   risk sharing problems studied in the literature, a risk sharing rule $\mathbf A$ does not require any information on the preferences
of the agents, a risk exchange market, or subjective decisions of a central planner.
The risk allocation will be determined completely through the design of the rule and the input risk vector.  


We next provide   several simple examples of risk sharing rules; see \cite{DDR22} for a collection of risk sharing rules and their properties. 
Throughout, for $q \in [0,\infty)$,   denote by $L^q = L^q(\Omega,\mathcal F, \p)$ the set of all random variables with a finite $q$-th moment, and let $L^q_+$ denote the set of non-negative elements of $L^q$. 
Let $L^\infty$ denote the set of essentially bounded random variables. 
We use the shorthand $L^q$, and we will write the full $L^q(\Omega,\mathcal G,Q)$ when we encounter another probability space $(\Omega,\mathcal G,Q)$.  
 Some of the risk sharing rules below require $\X$ to be a  subset of some specific  spaces. We always use the convention $0/0=0$ which may appear in degenerate cases of (vi) and (vii). 
\begin{enumerate}
	\item[(i)] The identity risk sharing rule 
	$
{\mathbf A}_{\rm id}^\XX
=\XX\mbox{~for~}\XX\in \X^n.
$
	 
	\item[(ii)] The all-in-one risk sharing rule $\mathbf A_{\rm all}^\XX=\left (S^\XX,0,\dots,0 \right)\mbox{~for~}\XX\in \X^n.$
	\item[(iii)] The mean-adjusted all-in-one risk sharing rule
	 $$\mathbf A_{\rm ma}^\XX=\left (S^\XX-\E[S^\XX], 0,\dots,0 \right)+\E[\XX]\mbox{~for~}\XX\in \X^n\subseteq (L^1)^n.$$ 
	\item[(iv)] The uniform risk sharing rule 
 $\mathbf A_{\rm unif}^\XX=S^\XX \left (\frac 1 n,\dots,\frac 1 n \right)$ for $\XX\in \X^n.$
	\item[(v)] The conditional mean risk sharing rule  (CMRS)	
 $	\mathbf A_{\rm cm}^\XX  =  \E \left[\mathbf X |S^\XX \right]$ for $\XX\in \X^n\subseteq (L^1)^n.
	$
	
	\item[(vi)] The mean proportional risk sharing rule (MPRS) 
 $$\mathbf{A}_{\rm prop}^\XX =\frac{ S^\XX}{\E[S^\XX]} \E[\XX]\mbox{~for~}\XX\in \X^n\subseteq (L^1_+)^n.$$
We  can also define $\mathbf{A}_{\rm prop}^\XX$ for $\XX\in (L^1)^n$ satisfying $\E[S^\XX]\ne 0$ as in Section \ref{sec:22} below. We keep the formal definition of  $\mathbf{A}_{\rm prop}$ on subsets of  $(L^1_+)^n$ so that $\X$ can be any convex cone   in $L^1_+$.
	\item[(vii)] The covariance risk sharing rule $$\mathbf{A}_{\rm cov}^\XX = \frac{S^\XX - \E[S^\XX]}{\var(S^\XX)}\cov(\XX, S^\XX)   +\E[\XX] \mbox{~~~for~}\XX\in \X^n\subseteq (L^2)^n.$$ For jointly normally distributed initial risks, this rule is the same as CMRS. 
\end{enumerate} 

Among these rules, the last three (v)-(vii) are popular in the recent risk sharing literature; see \cite{DDR22}. CMRS in (v) is the most important for our  theory of anonymized risk sharing.  
These rules are connected and sometimes coincide. For instance, 
CMRS reduces to MPRS in the particular cases studied by \cite{FKZ18},
and (v)-(vii) all reduce to the uniform risk sharing rule (iv) when  $\mathbf X$ is exchangeable.


\subsection{An Illustrative Example}
\label{sec:22}

To motivate the study of different risk sharing rules, consider a simple example of independent normal initial risks.  Suppose that there are three participants
in a P2P insurance network whose losses are given by ${\mathbf X}=(X_1,X_2,X_3)$ where $X_i\sim 
\mathrm{N}(\mu_i, \sigma_i^2)$, $i=1,2,3$, are independent. 
Let $S=X_1+X_2+X_3$. 
 The risk sharing rule  $A$ gives the allocation  $(A^{\mathbf X}_1,A^{\mathbf X}_2,A^{\mathbf X}_3)$
satisfying $A^{\mathbf X}_1+A^{\mathbf X}_2+A^{\mathbf X}_3=S$.
\begin{enumerate}
\item[(a)] The CMRS rule gives  $A_i^{\mathbf X}=\E[X_i|S]=\mu_i+\frac{\sigma_i^2}{\sigma_1^2+\sigma_2^2+\sigma_3^2}(S-(\mu_1+\mu_2+\mu_3))$ for $i=1,2,3$.
\item[(b)]   The MPRS rule gives 
$A_i^{\mathbf X}=\frac{\mu_i }{\mu_1+\mu_2+\mu_3}S$ for $i=1,2,3$.
 \end{enumerate}
 As we can see, the CMRS and   MPRS  rules give quite different allocations. For a concrete comparison, suppose that $\mu_1=5$, 
$\mu_2=4$, $\mu_3=3$, 
$\sigma_1^2=1$, $\sigma_2^2=4$ and $\sigma_3^2=3$.
Under CMRS, we get $$A_1^{\mathbf X}=\frac{S}{8}+\frac{7}{2};~~~A_2^{\mathbf X}=\frac{S}{2}-2;~~~  A_3^{\mathbf X}=\frac{3S}{8}-\frac{3}{2}.$$
It is straightforward to check that $\E[A^{\mathbf X}_i]=\E[X_i],$ and $\var(A^{\mathbf X}_i)<\var(X_i)$ for $i=1,2,3$. In general, CMRS is mutually beneficial for all participants to reduce their variances; a stronger version of this property, called universal improvement, will be studied in Section \ref{sec:UI} in detail.
On the other hand, under MPRS, we get
$$A_1^{\mathbf X}=\frac{ 5S}{12};~~~A_2^{\mathbf X}=\frac{S}{3};~~~ A^{\mathbf X}_3=\frac{ S}{4}.$$
It holds true that $\E[A^{\mathbf X}_i]=\E[X_i]$ for $i=1,2,3$,
but $\var(A^{\mathbf X}_1)=25/18>1 = \var(X_1)$; that is, the first agent has a larger variance after risk sharing, possibly jeopardizing the agent's willingness to participate. 
For this reason, CMRS may be preferable. 

A potential drawback of the CMRS occurs when the realized value of $S$ is nonpositive, i.e., no aggregate loss occurs, which happens with a relatively small probability as $\p(S\le 0)\approx 10^{-5}$. For instance, if $S=0$, then there is still a monetary transfer between agents under CMRS as $A^{\mathbf X}=(7/2,-2,-3/2)$, whereas for MPRS there is no monetary transfer.  

In summary, different rules may serve different purposes. Our axiomatic study below helps to understand which ones are desirable in the context of anonymized risk sharing.

\section{Four Axioms for Anonymized Risk Sharing}
\label{sec:ax}
 
We next discuss desirable axioms  for risk sharing rules by addressing the considerations of both \emph{fairness} and \emph{anonymity}. 
Fairness refers to the feature that each agent does not receive an absurd or unjustified allocation, and anonymity refers to the feature that agents do not need to disclose information on their identity,   wealth,    preferences, private operations,  and final realized losses. 
Since anonymity also guarantees that each agent will not be treated differently and reduces discrimination, anonymity  and fairness are clearly related, although the two concepts have different motivations. 
We propose four axioms on a desirable risk sharing rule $\mathbf A$. Two of these axioms may be categorized as fairness axioms, and the other two as anonymity axioms.

\begin{axiom}{AF}{Actuarial fairness}
The expected value of each agent's allocation coincides with the expected value of the initial risk. More precisely, $\E[\mathbf A^{\XX}] = \E[\XX]$ for $\XX\in \X^n$.
\end{axiom}

Axiom AF reflects one of the oldest and most fundamental principles in risk
management, which dates back to at least the 16th century; see, e.g.,  \cite{HPT20} and \cite{F20}.  
AF serves as the basis for premium pricing  in insurance, and this served as one of the earliest sources for studying probability and statistics.
Certainly, not all risk exchanges in practice are actuarially fair. 
In our framework, 
because no information about agents' preferences or identities is available, it should not happen that one agent would receive an allocation with a higher expected value than her contribution, and some others receive allocations with lower expected values. Recall that the sum of these expected values is equal to the sum of the total risk, and hence  agents on average  receive the same expected value before and after risk sharing. 
Based on the above reasons, 
 AF is a natural requirement for anonymized risk sharing. 

Axiom AF can be alternatively formulated via incentive compatibility of certain agents in the risk sharing pool, assuming that all agents agree on the probability $\p$.
Recall that a risk-neutral agent would not join the risk sharing pool if the expected value of their loss increases after the risk exchange.
If AF fails, then some agents will have a higher expected loss.
A risk sharing rule should not exclude, by design, risk-neutral agents;\footnote{Certainly,  the design of any risk sharing mechanism excludes some types of agents; however, it seems that excluding risk-neutral agents would be quite undesirable.} recall that the fundamental model of insurance in \cite{A63} involves a risk-neutral insurer to help share losses from risk-averse insureds. 

Axiom AF requires a pre-specified probability $\p$ with respect to which the expected value is computed.
As we will see below, none of the other axioms use the probability $\p$ directly. 
Therefore, $\p$ can be chosen freely in different contexts. 
In practice, objective and subjective prices co-exist in decision making (see e.g., \cite{GMMS10}).
 In our setting, this probability depends on the designer of the rule (e.g., an exchange platform or a blockchain),
and it may be either an objective or a subjective choice.  
  An example is that $\p$ may represent a (frictionless) pricing measure in a  financial exchange market. In this case,  
  the interpretation of AF is that each agent's allocated position has the same price as their initial endowment, which is satisfied at market equilibrium.
  Section \ref{sec:73} contains a case study in P2P health care insurance, where $\p$ is chosen differently from the estimated objective probability. 
  
\begin{axiom}{RF}{Risk fairness}
    The allocation to each agent should not exceed their maximum possible loss.   That is, for $\XX\in \X^n$ and $i\in [n]$, 
we have $ A^{\XX}_i \le \sup X_i$.
\end{axiom}

Axiom RF dictates that when agents join the pool to share their risk, 
they should not have to suffer more than their worst-case loss.\footnote{We can alternatively formulate Axiom RF using $ A^{\XX}_i \ge \inf X_i$, where $\inf$ is the essential infimum,
yielding the same interpretation if we interpret positive values of random variables as gains instead of losses. The mathematical results in the paper for this formulation remain the same due to symmetry.}  
For each realization of the actual losses, the allocation satisfies the no rip-off condition in the insurance pricing literature (\cite{DG85}), which says that an insured will never pay more premium than their maximum possible loss. 
If $ A^{\XX}_i \le \sup X_i$ does not hold, then an agent with no risk of default  may 
introduce positive probability of default after risk sharing, a clearly undesirable situation. 
For instance, using a power or logarithmic utility function, an agent's potential loss should never exceed her total wealth level (this may be called \emph{bankruptcy aversion}), and Axiom RF says that there is no bankruptcy after the risk exchange if the initial risk is safe in this regard.  
Formulated via incentive compatibility, this axiom means that not all bankruptcy-averse agents are excluded, arguably a weak requirement.

Two special   implications of RF may be useful. 
First, if an agent brings a pure surplus to the pool, i.e., $X_i\le0$, its allocation should also be a pure surplus; this is certainly true if $\X$ is contained in a half space such as the space of negative random variables.
Second, in conjunction with Axiom AF, RF yields for $\XX\in \X^n$ and $i\in [n]$, 
\begin{equation}
    \label{eq:constancy} 
    \mbox{if $X_i=x$ is a constant, then $A_i^\XX=x$;}
\end{equation}
this follows from $A_i^\XX\le x$ and $\E[A_i^\XX]=x$.  That is, if the initial risk of an agent is a  constant, then there is no risk exchange for this agent; this is quite intuitive since any risk-averse agent (in the sense of \cite{RS70}) would not trade a constant risk with a non-constant risk with the same mean. We emphasize \eqref{eq:constancy} by calling it Property CP. 
\begin{propty}{CP}{Constant preserving}
A constant initial risk results in
constant allocation. That is, for $\mathbf X\in\mathcal X^n$ and $i\in[n]$,
$X_i=x\in\mathbb R$ implies $A_i^{\mathbf X}=x$.
\end{propty}
From the perspective of the risk sharing platform, an agent without risky endowment is considered as not participating. 
In particular, for a risk vector $(X,0,\dots,0)$, i.e., only the first agent having a non-zero initial risk, this implies 
\begin{equation}
\label{eq:trivial}
A_1^{(X,0,\dots,0)} = X \mbox{~and~} 
A_j^{(X,0,\dots,0)} = 0 \mbox{~for~$j\ne 1$},
\end{equation}
which is arguably the only reasonable allocation in this particular case. On a point related to \eqref{eq:constancy} and \eqref{eq:trivial}, our framework does not include the mechanism of side-payments, as in e.g., selling insurance, because deciding side-payments usually requires the knowledge of specific identities or preferences (e.g., which agent is institutional, more risk averse, or with more bargaining power), and they are not done through the risk sharing platform.

Axioms AF and RF are both implied by a stronger property, called universal improvement (UI), which means 
that all risk-averse agents prefer the allocated risk to their initial risk. 
We choose to present our main results with AF and RF, which is a stronger axiomatic characterization. 
Results based on UI are presented in Section \ref{sec:UI}.
 
\begin{axiom}{RA}{Risk anonymity}
    The realized value of the allocation to each agent is determined by  that of the total risk. That is, for   $\XX \in \X^n$, $\mathbf A^\XX$ is $\sigma(S^\XX)$-measurable, where  $\sigma(S)$ is the $\sigma$-field generated by $S\in \X$.
\end{axiom} 

Axiom RA is central to the idea of designing a risk sharing mechanism. 
It means that the total realized allocation is  determined only by the total loss suffered by the risk sharing pool, and not by specific losses from the individual participants. This resembles  the earliest idea in insurance and risk sharing:  Individuals get together to share their total future losses (in early years, these losses are typically caused by unexpected deaths, diseases or injuries), regardless of which one of them is the realized cause of the future loss. In other words, once an agent enters the pool, her own realized loss no longer matters, and only the  realized loss  of the pool matters. 
This reflects a consideration of anonymity, as each agent does not need to disclose what is the realized loss; all individual losses are masked and only the total loss is revealed.  
The knowledge of the initial risk vector is only used for the design of the risk sharing mechanism, but not for the 
settlement of actual losses (see Figure \ref{fig:risk}).

Axiom RA has an interpretation via Pareto optimality (see Section \ref{sec:equ} for a precise definition).
For any allocation $(Y_1,\dots,Y_n)\in \mathbb A_n(S)$,
the random vector $(\E[Y_1|S],\dots,  \E[Y_n|S])$ is an allocation of $S$.
Any risk-averse agent would prefer $\E[Y_i|S]$ to $Y_i$, as shown by \cite{RS70}.
Therefore, any allocation $(Y_1,\dots,Y_n)$ of $S$, if not $\sigma(S)$-measurable, can be improved to be $\sigma(S)$-measurable, and this improvement is strict if agents are strictly risk averse. 
In other words, Axiom RA is a necessary (but not sufficient) condition for Pareto optimality of the allocation among strictly risk-averse agents with unspecified preferences.     
Moreover, RA holds for all risk sharing rules
satisfying that $\mathbf{A}^{\mathbf{X}}$ is comonotonic; 
a random vector $(X_1,\dots,X_n)$ is \emph{comonotonic} if 
there exist increasing (in the non-strict sense) functions $g_1,\dots,g_n$ 
and a random variable $Z$ 
such that $X_i=g_i(Z)$ a.s.~for $i\in [n]$.  Comonotonicity is closely related to Pareto optimality for  risk-averse agents.
 The above points will be made precise in Section \ref{sec:equ}.

\begin{axiom}{OA}{Operational anonymity}
    The allocation to one agent is not affected if risks of two other agents merge. That is, $A_k^{\mathbf{Y}} = A_k^{\mathbf{X}}$ for $k \neq i, j$  for $\mathbf X \in \X^n$ and distinct $i,j\in [n]$ and $\mathbf{Y}  = \mathbf{X}+ X_j \mathbf{e}_i - X_j \mathbf{e}_j $, where $\mathbf{e}_k = (0, \cdots, 0, 1, 0, \cdots, 0)$ is the unit vector along the $k$-th axis (the $k$-th component is 1).
\end{axiom}



 In the definition of Axiom OA, the risk vector $\mathbf{Y}$ can be expressed by $ Y_i=X_i+X_j$, $Y_j=0$ and $Y_k=X_k$ for $k\ne i,j$.
Axiom OA means that merging the risks of two agents will not affect the allocated components of uninvolved agents.
This also implies that a redistribution of risks between distinct agents $i$ and $j$ does not affect agent $k$ for $k\ne i,j$.
In an anonymized risk sharing framework, two agents may be two different accounts of the  same family, same organization, or even the same person. Their internal (private) operations do not need to be disclosed and should not affect the allocation to other agents.
OA is closely related to the fair-merging property of \cite{DDR22}, where there are $n-1$ agents after the merger, whereas here there are still $n$ agents. 
In the context of Bitcoin mining, \cite{LS20} formulated two axioms, called 
robustness to Sybil attacks and robustness to  merging,
which together reflect a similar motivation as OA.  The axiom OA is further explained in the following simple example.

\begin{example}\label{ex:2} 
	Assume $\XX = (X_1, X_2, X_3)$ and $\mathbf{Y} =  (X_1 ,X_2+X_3,0)$. In this setting, we have $A_1^\mathbf{Y} = A_1^{\XX}$ if Axiom OA holds. Further, we have $A_3^{\mathbf{Y}}  = 0$ from \eqref{eq:constancy} implied by AF and RF,  leading to $A_2^{\mathbf{Y}} = A_2^{\XX} + A_3^{\XX}$.  
Therefore, by merging risks from agents $2$ and $3$,   agent $2$ now takes up the total allocation to the two agents, and the allocation to agent $1$ is unaffected by this operation. Furthermore, the variance sharing rule, i.e., $A_i^\mathbf{X} =\frac{\var(X_i)}{\sum_{i=1}^n \var(X_j)} S^\XX$ does not satisfy Axiom OA, unless $X_j$'s  are mutually independent.
\end{example}

Axiom OA can be alternatively formulated by another intuitive property that $A_i^\XX$ is determined by $(X_i,S^\XX)$ for each $i$ and $\XX$. 
This latter property implies OA by definition.
To see that OA implies this property,  it suffices to observe, by repeatedly merging all agents except for agent $1$, that 
 $
A_1^{\XX} = A_1^{(X_1,S^\XX-X_1,0,\dots,0)}
$ holds.  
We summarize the above observation in the following proposition.

\begin{proposition}\label{prop:OA}
A risk sharing rule $\mathbf A$ satisfies 
Axiom OA  if and only if for all $\XX\in \X^n$ and $i\in [n]$, $A_i^\XX$ is completely determined 
by the random pair $(X_i,S^\XX)$ (i.e., there exists a mapping $B_i$ such that
$
A_i^{\mathbf{X}}=B_i(X_i,S^{\mathbf{X}})
$ 
for all $\mathbf{X}\in\mathcal X^n$ and all $i \in [n]$). 
\end{proposition}



\section{An Axiomatic Characterization} 
\label{sec:axiom} 
In this section, we show   that  Axioms AF, RF, RA and OA uniquely identify  CMRS among all risk sharing rules. 
 Recall that   CMRS is defined  as
$$	\mathbf A_{\rm cm}^\XX 
  =  \E \left[\mathbf X |S^\XX \right] \mbox{~~~for~}\XX\in \X^n\subseteq (L^1)^n.$$
For the ease of presentation, we take $\X=L^1$ or $L_+^1$ in all our results;  these  results hold true for  $\X=L^q$ and $\X=L_+^q$ with $q\in [1,\infty]$  following the same proof; see Remark \ref{rem:Lq} in Appendix \ref{app:pf4}.

We first briefly check that CMRS satisfies the four axioms by using properties of the conditional expectation $\E[X|S]$ for any $(X,S)\in\X^2$ which will be chosen as $(X_i,S^\XX)$ for $i\in [n]$. First, AF holds by the tower property $\E[\E[X|S]]=\E[X]$.
Second,  RF holds since $\E[X|S]\le \sup X$.
Third, RA holds by definition since $\E[X|S]$ is a function of $S$.  
Fourth, OA holds since $\E[X|S]$ is determined solely by $(X,S)$. 
See also \cite{DDR22} for these and other properties of CMRS.

\begin{theorem} \label{th:full}  	 
	Assume   $\X= L^1$ or $\X=L^1_+$. A risk sharing rule satisfies Axioms AF, RF, RA and OA {if and only if} it is CMRS.
\end{theorem}

Theorem \ref{th:full} is the main result of the paper, showing that the four  fairness and anonymity axioms  allow for only one risk sharing rule. 
To the best of our knowledge, Theorem \ref{th:full} is the first axiomatic characterization of   risk sharing rules (not just anonymized risk sharing rules) in the literature.

The ``if" statement in Theorem \ref{th:full}, that CMRS satisfies the four axioms, has been checked above. The ``only if" statement, which is the most important part of  Theorem \ref{th:full}, requires a much more involved proof based on some advanced  results from functional analysis. 
In Appendix \ref{sec:AppA} we explain the proof  in the special case that $(\Omega,\mathcal F,\p)$ is discrete and $\X$ is the set of all random variables on this probability space. 
This proof heavily relies on the assumption that $\p(S=s)>0$ for any possible value $s$, and it cannot be directly generalized to non-discrete spaces using continuity arguments. 
For a   proof of Theorem \ref{th:full} on general probability spaces, 
we need a more refined representation result in functional analysis in Theorem \ref{thm:BT} below. 
\begin{theorem}\label{thm:BT}
For a random variable $S \in L^1(\Omega,\mathcal F,\p)$,  
 consider the mapping $\phi: L^1(\Omega, \mathcal F, \p)\to L^1(\Omega, \sigma(S), \p)$. The mapping
$\phi$ is the operator induced by the conditional expectation, 
i.e.,  $\phi(X)=\E[X|S]$ as  a mapping $X \to \E[X|S]$ for $X\in L^1(\Omega,\mathcal F, \p)$, 
if and only if $\phi$ satisfies the following properties: (a) $\phi(t)=t$ for all $t\in \R$;
	(b) $\phi(X+Y) = \phi(X) + \phi(Y) $ for all $X,Y \in L^1(\Omega,\mathcal F,\p)$;
	(c)  $\phi(Y) \ge \phi(X)$ if $Y\ge X$ and $X,Y \in L^1(\Omega,\mathcal F,\p)$;
	(d) $\phi(S)=S$, and
	(e) $\E[\phi(X)]=\E[X]$ for all $X \in L^1(\Omega,\mathcal F,\p)$. 
\end{theorem} 

The $\mathcal G$-conditional expectation as a mapping from  $L^1(\Omega, \mathcal F, \p)$ to $L^1(\Omega, \mathcal G, \p)$ admits a few different sets of characterizations; \cite{P67} has a collection of several early results. For a more recent account, see \citet[Chapter 13]{EFHN15} in the context of Markov operators. Theorem \ref{thm:BT} extends the above literature by offering a new characterization of the conditional expectation.

There is a substantial gap between characterizing risk sharing rules, which is absent from the literature, and characterizing the conditional expectation. 
There are three key differences. First,
our framework concerns risk sharing rules, which are mappings from $
\X^n$ to $\X^n$  with the additional property of full allocation, whereas  the literature on the conditional expectation concerns mappings from $\X$ to $\X$. Second, the property $\phi\circ \phi=\phi$, called 
idempotence or projection (see e.g., \cite{D65} and \cite{P67}), is essential  to characterize the conditional expectation, whereas in our framework this property is implied as a consequence of the four axioms on the risk sharing rule. 
Third, in our framework, the output $\sigma$-field $\mathcal G=\sigma(S)$ depends on the input random vector $\XX$, whereas for results on the conditional expectation, the $\sigma$-field $\mathcal G$ is pre-specified.

Theorem \ref{th:full} treats both the spaces $\X=L^1$ and $\X=L^1_+$ (the space of non-positive random variables is similar).
These two cases  represent different contexts. 
In some applications, risk allocation and risk contributions are restricted to being all positive or all negative, depending on the context. For instance, if agents are sharing P2P insurance losses, then it may be sensible to assume $\XX$ and $\mathbf{A}^\XX$ both take non-negative vector values; if agents are sharing profits from an investment, then it is the opposite; recall that positive values represent losses and negative values represent gains. 

In the next proposition, we verify that the four axioms are independent in the sense that  none of them can be removed from Theorem \ref{th:full}. 

\begin{proposition}\label{prop:indep}
Assume $\mathcal{X}=L^1$ and the probability space is atomless. None of Axioms AF, RF, RA and OA is redundant. That is, any combination of three of Axioms AF, RF, RA and OA does not imply the remaining fourth axiom.
\end{proposition}

To prove Proposition \ref{prop:indep}, for each axiom, we will provide in Example \ref{ex:indep} a risk sharing rule that only satisfies three of them but not the fourth one; some of these examples have been listed in Section \ref{sec:2}.  The technical details of these claims are in Appendix \ref{app:pf4}.

\begin{example}
\label{ex:indep}
\begin{enumerate}
	\item[(i)] The $Q$-CMRS   $\mathbf{A}_{Q\text{-}{\rm cm}}^\XX = \E^Q[\XX|S^\XX]$ 
    on $\mathcal X=L^1(\Omega,\mathcal F,\mathbb P)$, where
$Q\sim\mathbb P$, $Q\ne\mathbb P$, and both $\mathrm{d} Q/\mathrm{d}  \mathbb{P}$ and
$\mathrm{d}  \mathbb{P}/\mathrm{d} Q$ are essentially bounded, satisfies RF, RA and OA, but not AF.
	\item[(ii)]
	The mean-adjusted all-in-one risk sharing rule
	$\mathbf A_{\rm ma} $ in Section \ref{sec:21} 
	with $\X\subseteq L^1$ satisfies AF, RA and OA, but not RF. As further examples, both the covariance risk sharing rule and  the mean proportional risk sharing rule   in Section \ref{sec:21}  satisfy AF, RA and OA, but not RF.	
	\item[(iii)] The identity risk sharing rule $\mathbf{A}_{\rm id}^\XX = \XX$ satisfies AF, RF and OA, but not RA.
	
	\item[(iv)] A combination of $\mathbf{A}_{\rm all}$ and $\mathbf{A}_{\rm cm}$, defined by
	$\mathbf{A}^\XX = \mathbf{A}_{\rm all}^\XX =  (S^\XX, 0, \cdots, 0)$ if $\mathbf{X}$ is an i.i.d. standard Gaussian random vector, 
	and $\mathbf{A}^\XX =\mathbf{A}_{\rm cm}^\XX = \E[\XX|S^\XX]$ otherwise, satisfies AF, RF and RA, but not OA, because the sum of two standard Gaussian random variables is not standard Gaussian.
\end{enumerate}
\end{example}

\section{Other Properties and Their Connection to The Four Axioms}
\label{sec:property}


\subsection{Universal Improvement in Terms of Convex Order}
\label{sec:UI}

For two random variables $X$ and $Y$, we say $X$ is improved compared to $Y$ in \emph{convex order} if $\E[v(X)]\le \E[v(Y)]$ for any convex function $v: \R \to \R$; this is denoted by $X \le_{\rm cx} Y$.
The most appealing feature of CMRS, as argued by \cite{LM94} and \cite{DD12}, is that it universally improves the risk in terms of the convex order.

\begin{propty}{UI}{Universal improvement in convex order}
    The allocation improves the initial risk in convex order. More precisely, for any $\XX\in \X^n$ and  $i\in[n]$, it holds that $A^\XX_i\le_{\rm cx} X_i$.
\end{propty}

CMRS satisfies UI as a direct result of conditional Jensen's inequality.  
Intuitively, UI means that the initial risk  for each agent has larger variability than the allocation to that agent. 
Convex order is used to define risk aversion (\cite{RS70}), and it can be characterized through preferences for insurance (\cite{MMWW25}). 
Thus,  risk-averse agents, i.e., those who prefer an improvement of convex order,
will prefer their UI allocations over their initial risks; this is further discussed in Section \ref{sec:equ}. 
On a related issue, \citet[Proposition 4.2]{DR20b} showed that  if   risks in the pool are independent then the CMRS allocation  improves  in convex order for each existing agent  when  the pool is enlarged. 

To illustrate the important role of UI for CMRS, 
we note that UI implies both  AF and RF, since $X\le_{\rm cx} Y$ implies   $X\le \sup X\le \sup Y$ and $\E[X]=\E[Y]$ for any random variables $X,Y\in L^1$. 
Combining this with Theorem \ref{th:full}, we immediately obtain  in Theorem \ref{cor:UI} another characterization of CMRS with AF and RF replaced by UI.
Recall that our characterization   in Theorem \ref{th:full} relies on the weaker axioms of AF and RF, and thus  the more important ``only if" statement  is stronger than that of Theorem \ref{cor:UI}. 

\begin{theorem}\label{cor:UI}    
	Property UI implies Axioms  AF and RF and Property CP. 
Thus, for $\X= L^1$ or $L^1_+$,
	a risk sharing rule satisfies Axioms RA and OA and  Property UI if and only if it is CMRS.
\end{theorem}
More implications of Theorem \ref{cor:UI} will be discussed in Section \ref{sec:equ}. 
%
Next, we provide an example showing that the condition   $n\ge 3$ which we assumed from the beginning is indispensable, and this remains true even if we further assume the stronger property of UI. The intuition is that
in case $n=2$, Axiom OA is empty since no merging operation is possible when one agent's risk is fixed.


\begin{example}\label{ex:1} 
	Let $n=2$. We design a risk sharing rule $\mathbf A$ which satisfies all of AF, RF, RA, OA and UI, but it is not CMRS. 
	Let $\mathbf A=\mathbf A_{\rm cm}$ for all $\XX\in \X^n$ except for a specific $\mathbf Y=(Y_1,Y_2)$, 
	which is given by 
	$$
      A^{\mathbf Y}_1 = \E[Y_1|S^{\mathbf Y}]+h (S^{\mathbf Y}) \mbox{~and~}      A^{\mathbf Y}_2 = \E[Y_2|S^{\mathbf Y}]-h (S^{\mathbf Y}),
	$$ where $h$ satisfies $\E[h (S^{\mathbf Y})]=0.$
	The intuition is that, if $h $ is sufficiently small and $\E[Y_i|S^{\mathbf Y}]$ is sufficiently different from $Y_i$, then $A^{\mathbf Y}_i\le_{\rm cx} Y_i$ still holds, thus satisfying RA, OA and UI.
	To make the example explicit, let us take $Y_1\sim \mathrm{N}(0,1)$ and  $Y_2\sim \mathrm{N}(0,2)$, and $Y_1,Y_2$ are independent. Let $h(s)= s/6$. Note that $S^{\mathbf Y} \sim \mathrm{N}(0,3)$. We can compute 
	$$ A^{\mathbf Y}_1  = \frac{1}{3}S^{\mathbf Y}  + \frac{1}{6}S^{\mathbf Y}  = \frac{1}{2}S^{\mathbf Y}  \sim \mathrm{N}(0, 0.75) \mbox{~~
	and   ~~}
	 A^{\mathbf Y}_2  = \frac{2}{3}S^{\mathbf Y}  - \frac{1}{6}S^{\mathbf Y}  =\frac{1}{2}S^{\mathbf Y} \sim \mathrm{N}(0, 0.75).
	$$ 
	For this particular $\mathbf Y$,  
	Axiom RA and Property UI (hence Axioms AF and RF) are satisfied. 
	Axiom OA holds trivially as its statement is empty. 
	Thus, $\mathbf A$ is not  CMRS  but it satisfies the four axioms and Property UI. 
\end{example}

\subsection{Backtracking}

The second property we discuss is the backtracking property, which means that if $S^\XX$ determines $\XX$, then $\mathbf A^\XX=\XX$, i.e., $\mathbf A$ coincides with the identity rule $\mathbf A_{\rm id}$, and thus there is no risk exchange. 
It is straightforward to verify that CMRS satisfies this property. 

\begin{propty}{BT}{Backtracking}
    For each $\XX \in \X^n$, if $\XX$ is $\sigma(S^\XX)$-measurable, then  $\mathbf A^\XX=\XX$.
\end{propty}

Property BT intuitively means that there is no risk sharing  effect if $S^\XX$ is too informative compared to the individual contributions.  
As a consequence, there are some situations, although perhaps rare, in which CMRS discourages some participants from entering the risk sharing pool, even if they bring in risks independent of the other participants.
 Theorem \ref{th:full} provides the additional insight that BT is unavoidable, given the four natural axioms of fairness and anonymity. 

Property BT is sometimes argued as an undesirable property; see  \cite{DHR22}.
We give a simple example  below for the purpose of discussion. 

\begin{example} Suppose that $X_1=\sigma_1 Y_1  $, 
$X_2=\sigma_2  Y_2 $, and
$X_3=\sigma_3  Y_3 $, 
where $Y_1,Y_2,Y_3$ are i.i.d. taking values in $\{0,\dots,9\}$, and $\sigma_1=1001$, $\sigma_2=1010$ and $\sigma_3=1100$. Note that $S^\XX$ uniquely determines the value of $(X_1,X_2,X_3)$ since the last three digits of $S^\XX$ are precisely $Y_3,Y_2,Y_1$. 
In this example, $X_1,X_2,X_3$ have similar distributions, and they are independent. Intuitively, some risk sharing effect is possible for such $\XX$, but  $\mathbf A_{\rm cm}^\XX=\XX$ due to the backtracking property.
\end{example}

If  some applications demand BT   to be avoided, then one has to relax some axioms. For this, one naturally wonders which of the four axioms are responsible for Property BT. 
The axiom which involves $\sigma(S^\XX)$ is RA, and a first guess may be that RA is connected to BT. 
Somewhat surprisingly,  this is not true. In the next result, we establish that AF, RF, and OA are sufficient for BT 
if $\mathbf{A}^\XX$ is further assumed $\sigma( \XX)$-measurable; the last assumption means that the risk settlement does not involve extra randomness outside $\sigma(\XX)$.

\begin{proposition}\label{prop:BT}
		Assume $\X = L^1 $ or $L_+^1$. If a risk sharing rule $\mathbf{A} $ satisfies Axioms AF, RF and OA, and $\mathbf{A}^\XX$ is $\sigma(\XX)$-measurable for all $\XX \in \X^n$, then it satisfies Property BT.
\end{proposition}
An example of a risk sharing rule satisfying  AF, RF and OA but not RA is the mixture $\mathbf A=\lambda \mathbf A_{\rm id} + (1-\lambda)\mathbf A_{\rm cm}$  for some $\lambda \in (0,1]$; such a rule satisfies BT.
On the other hand, the mean-adjusted all-in-one and covariance risk sharing rules in Example \ref{ex:indep} satisfy 
AF, RA and OA, and they do not satisfy BT or RF. 






\subsection{Symmetry}\label{sec:54}

Symmetry is another important property reflecting the spirit of both fairness and anonymity. 
Let $\Pi_n$ be the set of $n$-permutations,
and we write $\XX_\pi=(X_{\pi(1)},\dots,X_{\pi(n)})$ for $\pi\in \Pi_n$ and $\XX \in \X^n$. 

\begin{propty}{SM}{Symmetry}
    For each $\XX\in \X^n$ and $\pi\in \Pi_n$, $(\mathbf A^\XX)_\pi =\mathbf A^{\XX_\pi}$.
\end{propty}

Property SM, also called the reshuffling property by \cite{DDR22}, reflects the consideration that if agents $i$ and $j$ exchange their initial risk contributions, then they also exchange their allocations.  Hence, their identities or positions in the risk sharing pool do not matter; this clearly relates to both fairness and anonymity. 
The principle that the positions of the agents are irrelevant is 
called anonymity in  some other contexts; see \cite{LS20} in Bitcoin mining reward sharing, and   
 \cite{DM08}  and \cite{BBLW16} in robust group decision rules.

 Property SM is not directly assumed among our axioms, and CMRS satisfies Property SM by definition. Therefore, SM must follow from some of the axioms we impose. 
 It turns out that OA and the following property are sufficient for SM. 
 \begin{propty}{ZP}{Zero preserving}
     For $\XX\in \X^n$ and $i\in [n]$, if $X_i=0$, then $A_i^\XX=0$.
 \end{propty}
 Recall  some relationship among some properties and axioms mentioned above, that is,
$$
\mbox{UI} \Longrightarrow \mbox{AF + RF} \Longrightarrow \mbox{CP} \Longrightarrow \mbox{ZP}.
$$

\begin{proposition}\label{prop:SM}
Axiom OA and Property ZP imply Property SM. 
\end{proposition}

SM does not come from any one of the four axioms alone. Indeed, the mean-adjusted all-in-one risk sharing rule in Example \ref{ex:indep} satisfies AF,  RA and OA, but not SM or ZP.
   The combination of $\mathbf{A}_{\rm all}$ and $\mathbf{A}_{\rm cm}$ in Example \ref{ex:indep}  satisfies AF, RF (hence ZP) and RA, but not SM or OA.

\section{Pareto Optimality and Comonotonicity}

\label{sec:equ}

The framework of anonymized risk sharing does not specify any preference
relations of the participating agents in the previous sections. 
Thus, without revealing preferences, a proposed anonymized risk sharing rule may not be optimal for given agents
with particular preferences.
However, we can discuss the Pareto optimality of allocation rules for some overall preference relations and some special individual preferences.

\subsection{Pareto Optimality}

We first fix some standard terminology in the analysis of optimality of allocations. Recall that we interpret positive values of $X \in \X$ as losses and negative values as gains. 
Each agent has a preference relation (a complete preorder) $\succeq$ on $\X$, with its strict relation denoted by $\succ$ and equivalence relation denoted by $\simeq$, satisfying that $X\ge  Y$ implies $Y\succeq X$. Moreover, we say that the agent is
\begin{enumerate} 
	\item[(i)] \emph{risk-averse} if 
	$X\ge_{\rm cx} Y$ implies  $Y\succeq X$;
	\item[(ii)] \emph{strictly risk-averse} if it is risk-averse and $X\ge_{\rm cx} Y$ and $X \stackrel{d}{\neq} Y$ imply  $Y \succ X$.
\end{enumerate}
We use $\succeq_i$ for the preference relation of agent $i \in[n]$.
An allocation $(Y_1,\dots,Y_n)\in\mathbb{A}_n(S)$ of $S$ is   \emph{Pareto optimal} if for any
$(Z_1,\ldots,Z_n)\in\mathbb{A}_n(S)$ satisfying $Z_i\succeq_i Y_i$ for all $i \in [n]$, we have $ Z_i\simeq_i Y_i$ for all $i\in[n]$. 
That is, one cannot find another allocation $(Z_1,\dots,Z_n)$ such that all agents are (weakly) better off, with at least one of them strictly better off.  
Pareto-optimal allocations are also known as Pareto{-efficient allocations.} 

For an initial risk vector $(X_1,\dots,X_n)\in\mathbb{A}_n(S)$, an allocation $(Y_1,\dots,Y_n)\in\mathbb{A}_n(S)$ is \emph{individually rational} if $Y_i\succeq_i X_i$ for all  $i=1,\dots,n$.
Intuitively, individual rationality means that the situation of each agent has either improved or remains indifferent after the allocation.
This condition is necessary for the participation of these agents in the risk sharing problem.  
Theorem \ref{cor:UI} justifies that CMRS is the only risk sharing rule that satisfies three economic principles formulated via   overall preferences of risk-averse agents. 

\begin{theorem}\label{prop:IR}
	Assume   $\X= L^1$ or $L^1_+$. 
	CMRS is the only risk sharing rule that satisfies, for any $n$ risk-averse agents, that
	\begin{enumerate}
	\item[{\rm [E1]}] the allocation is   individually rational;
	\item[{\rm [E2]}] the allocation cannot be strictly Pareto improved  by taking a common conditional expectation;\footnote{
    Here, for a risk sharing rule $\mathbf A$, [E2] means that for any $\mathbf{X}$, there does not exist a $\sigma$-field $\mathcal{G}$ with $\sigma(S^{\mathbf{X}})\subseteq\mathcal{G}$ such that $\mathbb{E}[\mathbf A^{\mathbf{X}}\mid\mathcal{G}]$ is weakly preferred to $\mathbf A^{\mathbf{X}}$ by all agents and strictly preferred by at least one agent. The strict part is understood over the class of strictly risk-averse agents.
}
	\item[{\rm [E3]}]  the allocation of each agent is unchanged when two other agents merge.
	\end{enumerate}
\end{theorem}

The proof of Theorem \ref{prop:IR} follows by first verifying that [E1], [E2] and [E3] are equivalent to Property UI and Axioms RA  and   OA, respectively, and then applying Theorem \ref{cor:UI}. 
Note that [E1], [E2] and [E3] are formulated via overall collections of risk-averse agents. Thus,  
Theorem \ref{prop:IR} highlights that our axioms for anonymized risk sharing treat agents as a whole, not individually.  

Next, we discuss some particular preferences and initial risk contributions 
for which CMRS indeed gives a Pareto-optimal allocation. 
An expected utility (EU) agent has preference $\succeq$ described by 
$
X \succeq Y \iff \E[u(-X)]\ge \E[u(-Y)]
$
for some increasing function $u:\R\to \R$, called the utility function.
An EU agent with utility function $u$ is risk averse if and only if $u$ is concave. 
If  $\XX$ is an exchangeable vector, then the CMRS allocation coincides with
the uniform allocation $(S^\XX/n,\dots,S^\XX/n)$, which  
is Pareto optimal for  EU agents with the same concave utility function $u$. 
This  optimality can be shown by noting that the  
concavity of $u$ yields 
$$
\sum_{i=1}^n u\left(-\frac{S^\XX}{n}\right) \ge \sum_{i=1}^n u(-Y_i) \mbox{ for any} (Y_1,\dots,Y_n)\in \mathbb A_n(S^\XX),
$$
which is a sufficient condition for Pareto optimality.
Finally, we mention a different form of optimality. 
As a property of the conditional expectation, when $X_i\in L^2$ for all $i\in[n]$, CMRS uniquely minimizes
\begin{equation}\label{eq:cost}
\E[(Y_i-X_i)^2] \mbox{~for all $i\in [n]$}
\end{equation}
over all choices of $(Y_1,\dots,Y_n)$ that are $\sigma(S^\XX)$-measurable;  this was observed by \citet[Proposition 4.4]{DD12}.
 Therefore, if we interpret \eqref{eq:cost} as an external transaction cost of settling the payment $Y_i-X_i$,
 then CMRS is optimal for minimizing such cost, among rules satisfying   [E2] above (or equivalently, Axiom RA).
Note that this minimality does not even need $(Y_1,\dots,Y_n)\in \mathbb A_n(S^\XX)$.

\subsection{Comonotonicity}

Pareto optimality for  risk-averse agents 
is connected to the concept of comonotonicity.

\begin{propty}{CM}{Comonotonicity}
The allocation $\mathbf A$ is comonotonic if, for each $\XX\in \X^n$, $\mathbf A^\XX$ is comonotonic.  
\end{propty}
Property CM implies Axiom RA since each component of a comonotonic random vector can be written as an increasing function of the sum;
see \cite{D94}. Therefore, Property CM can be equivalently formulated as that for $\XX\in \X^n$ and $i \in [n]$,  
there exists an increasing function $g_i^\XX: \R \to \R$ such that $A_i^\XX = g_i^\XX(S^\XX)$.
If Property CM holds, then the  allocation to  each agent increases as the total realized risk increases. 

Comonotonicity is sufficient for   Pareto optimality in some specific settings.
First, as a simple  example, any comonotonic allocation is Pareto optimal for risk-averse agents described by the same dual utility of \cite{Y87}; see Example \ref{ex:Yaari} in Appendix \ref{app:4}.
Second, the Pareto optimality can also be defined with respect to partial orders instead of complete orders. 
A popular choice is to use $\succeq$ as the increasing concave order, which represents the common preferences of all risk-averse EU agents; more precisely, $X\succeq Y$ if and only if $\E[u(-X)]\ge \E[u(-Y)]$ for all increasing concave functions $u:\R\to \R$ (Equivalently, $-X$ dominates $-Y$ in increasing concave order).\footnote{ 
For this choice of $\succeq$, 
an allocation $(Y_1,\dots,Y_n)\in\mathbb{A}_n(S)$ of $S$ is   {Pareto optimal} if for any
$(Z_1,\ldots,Z_n)\in\mathbb{A}_n(S)$ satisfying $Z_i\succeq Y_i$ for all $i \in [n]$, we have $ Z_i\simeq Y_i$ for all $i\in[n]$.}
Using the increasing concave order assumption, 
any comonotonic allocation is Pareto optimal (Theorem 3.2 of \cite{CDG12}).


Under belief homogeneity and mild assumptions, comonotonicity is also necessary for an allocation to be Pareto optimal for strictly risk-averse agents 
or to be a competitive equilibrium in  an exchange market with linear prices; see e.g., \cite{B62}, \cite{LM94} and \cite{BLW21}. 
However, as notable exceptions, both quantile-based risk sharing and belief heterogeneity can lead to non-comonotonic allocations that 
are still Pareto optimal; see \cite{ELW18, ELMW20}.

CMRS does not generally satisfy Property CM.
Thus, Theorem \ref{th:full} shows that the four axioms and Property CM conflict with each other. 
We further strengthen this result by showing that 
OA, CM, and ZP cannot be satisfied with the same risk sharing rule.  

\begin{proposition}\label{coro:conflict}  
Assume that $\X=L^1$ and 
the probability space supports a nonconstant integrable random variable. 
There is no risk sharing rule satisfying  Axiom  OA and Properties CM and ZP. 
\end{proposition}

There are two ways to mitigate the conflict between CMRS and comonotonicity.
First, there are specific models of $\XX$ for which comonotonicity of CMRS holds; see 
\cite{DDR22} and the references therein. For example, if $s\mapsto \E[X_i|S=s]$ is increasing for each $i\in [n]$, then CMRS is comonotonic. 
Second, if Property CM is needed, one may
  first apply CMRS to obtain the allocation $\mathbf{A}_{\rm cm}^\XX= \E[\XX|S^\XX]$ 
  and then use the comonotonic improvement of \cite{LM94} to give a comonotonic vector $\XX'\in \mathbb A_n(S^\XX)$ such that each component of $\XX'$ is dominated in convex order by the corresponding component of $A_{\rm cm}^{\mathbf X}$, and hence also by the corresponding
component of $\mathbf X$; see   \cite{R13} for this result on general spaces.
  Since any comonotonic allocation is Pareto optimal for the increasing concave order, 
  this combined comonotonic rule satisfies both UI and RA. Hence, unfortunately, the combined rule cannot satisfy OA by Theorem \ref{prop:IR}.  

\section{Link with Competitive Equilibrium}
\label{sec:equilibrium}

The axiomatic approach adopted thus far resides within the realm of normative economics; it prescribes how risk sharing ought to function on a decentralized platform based on a set of ideal axioms. This section bridges the gap to positive economics by investigating how these risk sharing rules can emerge endogenously in a decentralized economy through a pricing mechanism in a competitive market. Thus, the CMRS rule is not only a subjective reflection of societal goals but can also be derived objectively as a natural outcome of Arrow-Debreu competitive equilibrium under certain conditions.

	Consider $n$ agents with a vector $\mathbf X=(X_{1}, \dots, X_{n})$  of  initial endowments (represented by random losses as in the previous sections), who can access a complete Arrow-Debreu market to trade any payment $W_{i}$ at a linear market price $H(W_{i})$. Each agent maximizes expected utility:
	\begin{equation*}
		\max_{W_{i}}\E[U_{i}(-Y_{i})], \quad \text{where} \quad Y_{i}=X_{i}+W_{i}-H(W_{i}).
	\end{equation*}
	Note that $Y_i$'s are also losses and assume that each utility function $U_i$ is strictly increasing, strictly concave, and continuously differentiable on its domain. 
	Under standard conditions of linear pricing, the market price $H(W_{i})$ for $W_{i}$ can be computed through a pricing kernel $\Psi$: $H(W)=\E^{\Psi}[W]=\E[\Psi W].$ 
	Let $S^{\mathbf X} = \sum_{i=1}^n X_i$ be the aggregate risk. At the Arrow-Debreu equilibrium, the pricing kernel $\Psi$ is chosen such that the market-clearing condition holds: 
	$
	\sum_{i=1}^{n}Y_{i}=\sum_{i=1}^{n}X_{i} = S^{\mathbf X},$ that is, $ \sum_{i=1}^{n}\left(W_{i}-\E^{\Psi }[W_{i}]\right) =0.
	$ 
	Without loss of generality we assume a zero interest rate, i.e., $\E[\Psi]=1$. The resulting equilibrium allocation defines an equilibrium risk sharing rule $\mathbf X\mapsto \mathbf Y^{\ast}$, where
	$
	Y_{i}^{\ast}=X_{i}+W_{i}^{\ast }-\E^{\Psi }[W_{i}^{\ast}], \; i \in [n]
	$
   for some specific utility functions. In particular, $\E^\Psi[\mathbf A^{\mathbf X} ]= \E^\Psi [\mathbf X] $.
	The following theorem provides a sufficient condition for the CMRS rule 
    to emerge as the equilibrium allocation; that is, $Y_i^* = \mathbb E[X_i | S^{\mathbf X}]$ holds for every $i$.
	
	

	\begin{theorem}
	\label{prop:general-cmrs}
		Let $\mathbf X=(X_1,\dots,X_n)$ be a vector of initial endowments with $X_i \in L^1$ for all $i$. 
        Assume each utility function $U_i$ is strictly increasing, strictly concave, and continuously differentiable on its domain, and there exist constants $\lambda_1,\dots,\lambda_n>0$ such that
		\begin{equation}
        \label{eq:first-order-equilibrium}
		\lambda_i U_i'\big(-\mathbb E[X_i | S^\mathbf{X}]\big) = \lambda_j U_j'\big(-\mathbb E[X_j  | S^{\mathbf X}]\big) ~~\text{a.s.,}~~  \forall~ i,j.
		\end{equation} 
        {
        Define (i) the right-hand side as $M(S^{\mathbf X})$, a positive 
 $\sigma(S^{\mathbf{X}})$-measurable random variable; and (ii) the state-price density (pricing kernel) as $\Psi = c M(S^{\mathbf X})$, where $c = 1/\E[M(S^{\mathbf X})] \in (0, \infty)$ is a normalizing constant. 
 Assume that an Arrow-Debreu equilibrium exists and that all relevant expectations are finite.} Then the unique equilibrium allocation is given by the CMRS allocation {
        with respect to both the pricing kernel expectation $\E^\Psi$ and the objective expectation $\E$}: 
        $$
        Y_i^* = {
        \E^\Psi[X_i | S^{\mathbf{X}}] = \E[X_i | S^{\mathbf{X}}]} \;\; \text{ for all } i. 
        $$
        Conversely, if 
        the random vector ${(\E[X_1 | S^{\mathbf{X}}], \dots, \E[X_n | S^{\mathbf{X}}])}$ is  comonotonic and has continuous marginal distributions, we can find $U_1,\dots,U_n$ such that \eqref{eq:first-order-equilibrium} holds; thus, 
        there exist utility functions $U_1,\dots,U_n$ such that the CMRS allocation is the unique equilibrium allocation. 
	\end{theorem}


	Theorem \ref{prop:general-cmrs} highlights the dual nature of CMRS: it is both a normative principle (satisfying desirable axioms of fairness and aggregation) and can be a positive outcome (emerging from market equilibrium). 
{
In the  complete Arrow-Debreu market,   
        $\E^\Psi[\mathbf A^{\mathbf X} ]= \E^\Psi [\mathbf X] $ holds at equilibrium due to the market mechanism.
        Although this property is generally different from 
           Axiom AF,  that is,  $
     \E[\mathbf A^{\mathbf X}] =  \E[ \mathbf X ]  
        $,
         both properties hold simultaneously for a CMRS equilibrium allocation, as we have seen in Theorem \ref{prop:general-cmrs}. } 
	

    We give examples satisfying the conditions in Theorem \ref{prop:general-cmrs} for agents with HARA (Hyperbolic Absolute Risk Aversion) utilities, a broad family that includes exponential, quadratic, logarithmic utilities and the power utilities with
	the same power index as special cases. 
    This family leads to the quota-share (linear) rule \eqref{eq:quota-share} in equilibrium (see \cite{Borch1968}, \cite{Wilson1968}, and \cite{Schumacher2018}):  
	\begin{equation}
		Y_{i}^{\ast }=a_{i}S^{\mathbf X}+b_{i},  \label{eq:quota-share}
	\end{equation}%
	such that $\sum_{i=1}^{n}a_{i}=1$ and $\sum_{i=1}^{n}b_{i}=0$. 

    To make the equilibrium allocation coincide with the conditional expectation $\E[X_i|S^{\mathbf{X}}]$, 
    we further suppose the initial endowments $(X_1, \ldots, X_n)$ are elliptically distributed, a family that includes multivariate normal, Student-$t$, and other spherically symmetric distributions. A key property of elliptical distributions is that 
the conditional expectation $\E[X_i | S^{\mathbf{X}}]$ is a linear function of $S^{\mathbf{X}}$, i.e., 
$$
\E[X_i|S^{\mathbf X}] = \E[X_i] + \frac{\text{Cov}(X_i, S^{\mathbf{X}})}{\text{Var}(S^{\mathbf{X}})} (S^{\mathbf{X}} - \E[S^{\mathbf{X}}]).
$$
For concreteness, we consider multivariate normal risks with a positive definite covariance matrix $\Sigma = (\sigma_{ij})_{n \times n}$
with $(\Sigma \mathbf{1})_i > 0,$
where  $\mathbf{1}$ is the $n$-dimensional vector of ones.
The conditional expectation is linear, that is, 
$$
\mathbb E[X_i | S^{\mathbf{X}}] = c_i S^{\mathbf{X}} + d_i, 
$$  
for some constants $c_i,d_i$
where the slope coefficient $c_i$ is given by $c_i =  \frac{\text{Cov}(X_i, S^{\mathbf{X}})}{\text{Var}(S^{\mathbf{X}})} = \frac{\sum_{j=1}^n \sigma_{ij}}{\sum_{k=1}^n \sum_{l=1}^n \sigma_{kl}} = \frac{(\Sigma \mathbf{1})_i}{\mathbf{1}^\top \Sigma \mathbf{1}}>0. $ The linearity of the conditional expectations ensures that the random vector $(\E[X_1|S^{\mathbf{X}}], \ldots, \E[X_n|S^{\mathbf{X}}])$ is strictly comonotonic, since each component is a strictly increasing function of $S^{\mathbf{X}}$ (assuming $c_i > 0$ for all $i$, which holds when $X_i$ and $S^{\mathbf{X}}$ are positively correlated). 

By the last statement of Theorem \ref{prop:general-cmrs}, there exist utility functions $U_1, \dots, U_n$ within the HARA family such that the CMRS allocation is the unique equilibrium allocation. We can construct such utilities explicitly. The first example is the exponential utility (\cite{Gerber1998}, Example 7): 
$$
U_i(x)=- \exp(- x/c_i), \;\; i = 1, \dots, n.
$$ 
To see this, note that 
    $U'_i(-\E[X_i|S^{\mathbf X}]) = (1/c_i) \exp(S^{\mathbf X}+d_i/c_i)$, and therefore   \eqref{eq:first-order-equilibrium} holds with $\lambda_i= c_i \exp(-d_i/c_i)$ and $\Psi =  \exp(S^{\mathbf{X}})/\E[\exp(S^{\mathbf{X}})]$. 
The second example, adapted from \cite{Gerber1998} (Example 8), concerns power utility. 
This specification requires $S^{\mathbf X}>s$ a.s. Equivalently, since the context is 
about losses, $-S^{\mathbf X}<-s$ a.s., where $-s>0$ denotes the level of total maximal 
satisfaction of the $n$ agents. Fix the power index $\alpha>0$. We use the utility function
\[
U_i(x)
=
\frac{(-s_i)^{\alpha+1}-(-s_i-x)^{\alpha+1}}
{(\alpha+1)(-s_i)^\alpha},
\quad x<-s_i,\quad i\in [n],
\]
where $
-s_i=-d_i-c_i s>0
$ 
is the satisfaction bound of agent $i$, chosen so that
$
\sum_{i=1}^n s_i=s.
$ 
{
Strictly speaking, the condition $S^{\mathbf X}>s$ a.s. excludes an unbounded normal 
aggregate risk. In applications, however, one may take the total satisfaction bound 
$-s$ sufficiently large so that $\mathbb{P}(S^{\mathbf X}\le s)$ is negligible under 
a normal approximation; alternatively, one may work with a suitably truncated normal 
model.} Similarly, we have
$$
U_i'\left(-\mathbb{E}[X_i\mid S^{\mathbf X}]\right)
=
\frac{\left((-s)-(-S^{\mathbf X})\right)^\alpha}
{\left(-\frac{d_i}{c_i}-s\right)^\alpha}.
$$
Therefore, \eqref{eq:first-order-equilibrium} holds with $
\lambda_i=
\left(-\frac{d_i}{c_i}-s\right)^\alpha
$ 
and 
$
\Psi
=
\frac{\left((-s)-(-S^{\mathbf X})\right)^\alpha}
{\mathbb{E}\left[\left((-s)-(-S^{\mathbf X})\right)^\alpha\right]}.
$ 
This example includes the quadratic-utility case when $\alpha=1$.

{
We also note that the quota-share form \eqref{eq:quota-share} may not hold for other
	utilities. We show this in Appendix EC.5 by an example of an incomplete market with limited securities available to trade and the quadratic utility (mean-variance). 
    The example additionally shows an interesting insight that, in some cases, the risk exposure after the equilibrium risk sharing may be higher than that without risk sharing. In
	other words, the financial innovation of introducing a security may amplify
	the speculation. This result complements the results in \cite{Simsek2013}, in
	which the main reason for amplification is heterogeneous beliefs about the risk, while in our
	example, agents may have the same belief and the amplification comes from
	market incompleteness.   
}

\section{Generalized Risk Sharing Rules with Target Information}
\label{sec:generalized}

In some applications, more information than simply the realized value of the total risk is observable, and one may wish to allocate risks according to such information; see Sections 
\ref{subsec:72}, \ref{subsec:73}, and \ref{sec:74} below for real-world examples.  This leads us to propose a generalization of  risk sharing rules.

Denote by $\Sigma$ a set of sub-$\sigma$-fields of $\mathcal F$. 
A \emph{generalized risk sharing rule} is a mapping $\widehat {\mathbf A}:(\X^n\times \Sigma) \to\X^n$ satisfying $\widehat{\mathbf{A}}^{\XX|\mathcal G}=(\widehat A_1^{\XX|\mathcal G},\dots,\widehat A_n^{\XX|\mathcal G})\in \mathbb A_n(S^\XX)$ for each $\mathbf X \in \X^n$
and $\mathcal G\in \Sigma$.
The input $\sigma$-field $\mathcal G$  represents the information used to determine the 
realized values of the allocation, called \emph{target information}. Note that $ \sum_{i=1}^n \widehat A_i ^{\XX|\mathcal G}  = S^\XX$ implies that the $\sigma$-field of $\widehat{\mathbf{A}}^{\XX|\mathcal G}$ must contain $\sigma(S^\XX)$ regardless of the choice of $\mathcal G$. There may not exist  $\widehat{\mathbf A}$ such that $\widehat{\mathbf A}^{\XX|\mathcal G}$ is $ \mathcal{G}$-measurable for every $\mathcal G\in \Sigma$ and every $\XX\in \X^n$, because $\mathcal G$ may not contain the information of $S^\XX$. 
To address this issue, we merge the information in $\sigma(S^\XX)$ into $\mathcal G$, and  denote by  $\mathcal G^\XX =\sigma(S^\XX,\mathcal{G})$  the $\sigma$-field generated by $S^\XX$ and $\mathcal G$.
Below, we present two properties describing how the information modeled by $\mathcal G $ and $\mathcal G ^\XX$ is used for the generalized risk sharing rule $\widehat {\mathbf A}$.

\begin{propty}{IA}{Information anonymity}
    For $\XX \in \X^n$ and $\mathcal G\in \Sigma$, $\widehat{\mathbf A}^{\XX|\mathcal G}$ is $\mathcal G^\XX$-measurable.
\end{propty}

Property IA  reflects the idea that the risk allocations may not be solely determined by the realized value of $S^\XX$ but also depend on other information represented by $\mathcal G$. 
Property IA is a generalization of Axiom RA.
Property IA  gives   $\mathcal G^\XX$-measurability instead of $ \mathcal{G}$-measurability. 

\begin{propty}{IB}{Information backtracking}
For each $\XX \in \X^n$ and $\mathcal G\in \Sigma$, if $\XX$ is $\mathcal G^\XX$-measurable, then  $\widehat{\mathbf A}^{\XX|\mathcal G}=\mathbf X$.
\end{propty}

Property IB is a generalization of Property BT. 
It reflects the consideration that getting an allocation determined by $\mathcal G^\XX$ is our target.\footnote{
For a given risk sharing rule $\mathbf A$, we can define $\widehat{\mathbf A}$ 
by $\widehat{\mathbf A}^{\XX|\mathcal G}= \mathbf A^\XX$ for each $\XX\in \X^n$ and $\mathcal G\in \Sigma$; that is,  the information $\mathcal G$ is ignored. 
In this case, $\mathbf A$ satisfies RA if and only if $\widehat{\mathbf A}$ satisfies IA.
For instance, the generalized risk sharing rule 
$\widehat{\mathbf A}$ defined 
by $\widehat{\mathbf A}^{\XX|\mathcal G}= \mathbf A_{\rm cm}^\XX$ for each $(\XX,\mathcal G)$
satisfies Property IA, but this case is not interesting for this section, as it ignores $\mathcal G$; in particular, Property IB excludes this case,
since $\widehat{\mathbf A}^{\XX|\sigma(\XX)}=\XX$ needs to hold by Property IB.} More precisely, no risk exchange happens if the initial risk is already determined by $\mathcal G^\XX$.
For a generalized risk sharing rule $\widehat{\mathbf A}$,
we say that it satisfies an axiom or property introduced for risk sharing rules,
if the mapping $\XX\mapsto \widehat{\mathbf A}^{\XX|\mathcal G}$ satisfies the corresponding axiom or property for each $\mathcal G \in \Sigma$.
The next result characterizes the \emph{generalized CMRS},  defined by
\begin{equation}
\label{eq:generalized-def}
	 \widehat{\mathbf A}^{\XX|\mathcal G}_{\rm cm}= \E\left [\XX | \mathcal{G}^\XX\right]\mbox{~~~for~$\XX \in \X^n$ and $\mathcal G\in \Sigma$,}
\end{equation}
	among all generalized risk sharing rules.
 
\begin{theorem}  \label{th:full3}  	  
	Assume $\X= L^1$ or $L_+^1$.  A generalized risk sharing rule  satisfies Axioms AF, RF  and OA and Properties IA and IB {if and only if}  it is the generalized CMRS.
\end{theorem}

  Theorem \ref{th:full3} indicates that, assuming   AF, RF, OA, IA and IB,  the only generalized risk sharing rule with a given target information $\mathcal G^\XX$ needs to be calculated based on the conditional expectation with respect to $\mathcal G^\XX$. This interpretation is similar  to the result of Theorem \ref{th:full}.  
 The generalized CMRS characterized in Theorem \ref{th:full3} will be useful as many practical applications involve allocations that are not solely determined by the total risk.  
 We discuss some of them in Section \ref{sec:app}.

\section{Applications}
\label{sec:app}


\subsection{Blockchain Mining Rewards}
\label{sec:71}

\subsubsection{A Single Bitcoin Mining Pool}
\label{subsec:71}
By the design of the Bitcoin protocol, a decentralized network of anonymous computers, commonly called  miners, 
tries to solve a computational puzzle, which we refer to as a block in the blockchain.
After a block in the Bitcoin blockchain is successfully mined by a randomly selected miner, 
a predetermined number of bitcoins is awarded to that miner. For more background on Bitcoin mining, including criticisms on its environmental and economic impact and the conflict between mining pools and decentralization, we refer to  \cite{CK17}, \cite{ES18} and \cite{LS20}.  It is not our intention to say whether mining pools are good or bad;  taking their existence as given, our focus is the design of   reward sharing mechanisms within mining pools.



Mining activities are risky, with a large monetary value of the reward and a small probability of success, for individual miners.
For this reason, mining pools are formed by groups of miners to share the risk. 
 Risk-averse miners always have incentives to join mining pools to improve their utility.  
 
Suppose that $n$ miners form a mining pool to share the possible reward from mining the next block.
Let the random variable $P> 0$ represent the monetary value of the next block at the time of solving the block. 
The miners' initial contribution vector is  $\XX=
P\cdot(\vmathbb{1}_{D_1},\dots,\vmathbb{1}_{D_n})$, where $D_i \subseteq \Omega$ is the event that miner $i$ successfully issues the next block,\footnote{All sets $D\subseteq \Omega$ that appear in this section are assumed  measurable, i.e., $D\in \mathcal F$.} and the probability $ \p(D_i)$ represents the computational contribution of the miner $i$,    measured by the number of hashes tried, divided by that of all miners in the world mining the block. 
Since the reward has a non-negative monetary value, we interpret positive values of allocations as rewards in this subsection, a different sign convention from the rest of the paper; this causes no technical problem as all our results are invariant with respect to their signs. 
 The events $D_1,\dots,D_n$ are mutually exclusive because at most one miner can issue the next block. We assume that 
these events are independent of the bitcoin price $P$ because $P$  is determined by  market activities and $D_1,\dots,D_n$ are determined by randomly trying solutions.
Denote by $D=\bigcup_{i=1}^n D_i$ the event that any miner from this pool issues the block. 

Our model addresses miners in mining pools.
This  setting is distinct from 
those of \cite{ES18} and \cite{LS20}, which focus on the reward mechanism of home miners, i.e., those who do not participate in mining pools.
Moreover,
in the setting of their papers, 
the reward to an individual miner is a binary 
outcome, characterized by the probability of receiving $1$ block, whereas in our framework of a mining pool, miners can receive non-binary outcomes, 
depending on the pool's performance and their shares in the pool.
Consequently, all axioms in \cite{LS20} are formulated in terms of probabilities,
while our axioms are formulated on monetary outcomes. 
This distinction is also helpful in 
 the application of mining in multiple blockchains in Section \ref{subsec:73}. 

To make our analysis of   reward sharing    rigorous, let $P\in \X $ be a fixed positive random variable, and denote by
 $\mathcal B_n$   the set of all possible reward vectors of the miners, that is,
$$\mathcal B_n=\{P\cdot (\vmathbb{1}_{D_1},\dots,\vmathbb{1}_{D_n}): D_1,\dots,D_n\subseteq \Omega \mbox{~are disjoint and independent of $P$}\}.$$
We assume that the probability space is rich enough so that a continuously distributed random variable independent of $P$ exists. 
A \emph{reward sharing rule} is  a mapping
$
\mathbf A: \mathcal B_n \to \X^n 
$ 
satisfying $\mathbf A^\XX \in \mathbb A_n(S^\XX)$ for each $\XX\in \mathcal B_n$
and   $A_i^\XX =A_j^\XX$ for $i,j\in [n]$ with $\p( D_i)=\p(D_j)$ where $\XX=P \cdot (\vmathbb{1}_{D_1},\dots,\vmathbb{1}_{D_n})$.  
The last requirement reflects that only the amount of computational contribution of each miner is supplied (instead of the specification of the events $D_i$ and $D_j$), as
assumed in \cite{ES18} and  \cite{LS20}.  
The four axioms of fairness and anonymity have natural interpretations and desirability in this setting.  

\begin{enumerate}[(i)]
\item Axiom AF means that no agent gets less (or gets more) than their initial contribution in expectation, a simple form of fairness among anonymous participants. 
\item  Axiom RF (with a sign flip, i.e., $A^\XX_i\ge \inf X_i$ for $i \in [n]$) means that any miner has a non-negative reward, since $\inf (P\vmathbb{1}_{D_i})=0$ if $\p(D_i)\in [0,1)$.
Note that the case $\p(D_i)=1$ is trivial since all other agents have $0$ contribution and $0$ reward (using AF), and agent $i$ receives the whole reward $P$.
\item  Axiom RA means that the reward does not depend on which miner issues the block.
If the block is issued by the pool, the rewards to miners depend on their computational contributions and the bitcoin price, but not the actual issuing miner. This feature is central to the idea of creating a mining pool and joining computational resources.
\item  Axiom OA means that the mechanism is safe against merging and Sybil attacks, i.e., creating multiple accounts of the same participant.
Recall that miners are represented by computers and online accounts, and merging, splitting, or creating them is not disclosed to other miners. Hence, such operations between some miners should not affect the reward to an uninvolved miner.
\end{enumerate} 

As anticipated by Theorem \ref{th:full}, 
the only reward sharing rule satisfying Axioms RA, RF, AF and OA should be CMRS. 
This is indeed true,   although a separate proof is needed,
as the set $\mathcal B_n$ is much smaller than $\X^n$ and it  is not closed under addition, 
preventing us from directly applying Theorem \ref{th:full} or its proof. As before and throughout, ratios with zero denominator are interpreted using the convention
$0/0=0$. 
 
\begin{proposition}\label{prop:bitcoin}
Assume $P\in \X=L^1$ and $P>0$.
	A reward sharing rule $
\mathbf A: \mathcal B_n \to \X^n 
$ satisfies Axioms RA, RF, AF and OA if and only if it is specified by 
\begin{equation}
A^\XX_i = \frac{\p(D_i)}{ \p(D )}P \vmathbb{1}_D,~~~  i\in [n], 
\label{eq:bitcoin}
\end{equation}
where $\XX= P(\vmathbb{1}_{D_1},\dots,\vmathbb{1}_{D_n})\in \mathcal B_n$. This rule is CMRS. 
\end{proposition}

The allocation \eqref{eq:bitcoin} is precisely the common practice in mining pools; 
see \cite{ES18} and \cite{LS20} in the special case $P=1$.  
The total value $P$  is shared proportionally to the computational contribution of each miner  if this mining pool successfully issues the block (i.e., $\vmathbb{1}_D=1$),
and the rewards are $0$ otherwise.  
This reward sharing rule is CMRS, since 
 $\E[P \vmathbb{1}_{D_i} |P \vmathbb{1}_D] = P \vmathbb{1}_D \E[\vmathbb{1}_{D_i} |  \vmathbb{1}_D] =  P \vmathbb{1}_D \p(D_i)/\p(D).$\footnote{In this simple setting, CMRS also coincides with the mean proportional risk sharing rule.}
An example is shown in Figure \ref{fig:bitcoin}  to illustrate how \eqref{eq:bitcoin} works for  three miners. Before joining a mining pool,   miner $1$ gets the reward $P$   if  she issues the block (purple area in Figure \ref{fig:bitcoin}) and otherwise she receives nothing. After  miner $1$ joins the   group, the reward for her will be $P\p(D_1)/\p(D)$ if any of the three miners issues the block (brown area in Figure \ref{fig:bitcoin}). 

The new insight offered by Proposition \ref{prop:bitcoin} is that the reward sharing rule \eqref{eq:bitcoin} is the unique possible mechanism if our four axioms are considered as desirable,
and thus they fully rationalize the choice of this mechanism in practice. 

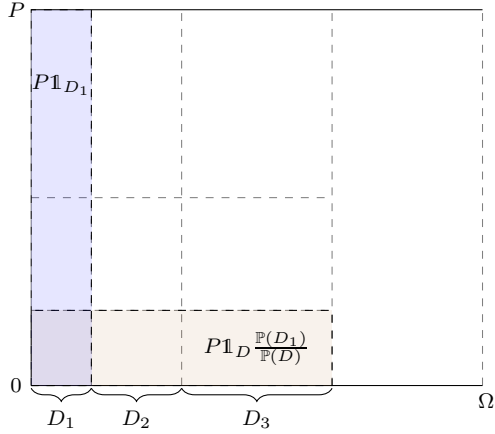
\begin{figure}[t]
	\caption{An illustration of a Bitcoin mining pool of 3 miners}
	\label{fig:bitcoin}
	\begin{center}
		\begin{tikzpicture} 
			\draw[gray,dashed] (0,0) -- (0,5);
			\draw[gray,dashed] (0.8,0) -- (0.8,5);
			\draw[gray,dashed] (2,0) -- (2,5); 
			\draw[gray,dashed] (4,0) -- (4,5); 
			\draw[gray,dashed] (0,1) -- (4,1);
			\draw[gray,dashed] (0,2.5) -- (4,2.5);  
			\draw[gray,dashed] (6,0) -- (6,5);  
			\draw [black] (0,0) -- (6,0);
			\draw [black] (0,5) -- (6,5);
			\node at (0.4,-0.45)   {\scriptsize $D_1$}  ;     
			\node at (1.4,-0.45)   {\scriptsize $D_2$}  ;   
			\node at (3,-0.45)   {\scriptsize$D_3$}  ;          
			\node at (-0.2,5) {\scriptsize $P$};  
			\node  at (-0.2,0) { \scriptsize$0$};       
			\node at (6.05,-0.2) { \scriptsize $\Omega$};  
			\draw [decorate,  decoration = {brace,    raise=1pt,    amplitude=5pt}] (0.8,0) --  (0,0);
			\draw [decorate,  decoration = {brace,    raise=1pt,    amplitude=5pt}] (2,0) --  (0.8,0);
			\draw [decorate,  decoration = {brace,    raise=1pt,    amplitude=5pt}] (4,0) --  (2,0); 
			\draw[   fill=brown, dashed,
			fill opacity = 0.1] (0,0) --  (0,1) -- (4,1) --  (4,0) -- (0,0);
			\draw[   fill=blue, dashed,
			fill opacity = 0.1] (0,0) --  (0.8,0) -- (0.8,5) --  (0,5) -- (0,0); 
			\node  at (3,0.5) { \scriptsize $P\vmathbb{1}_{D}\frac{\p(D_1)}{\p(D)}$}; 
			\node  at (0.4,4) { \scriptsize $P\vmathbb{1}_{D_1}$};     
		\end{tikzpicture} 
	\end{center}
\end{figure}

\subsubsection{Multiple Mining Pools} 
\label{subsec:72}

Next, suppose that miners can choose to participate in multiple pools by allocating their computational resources among these pools.
We would like to pin down a suitable allocation rule in this setting with the help of the generalized CMRS and a specific choice $\mathcal G$ of  target information.
 
There are $m$ mining pools.
Let $E_1,\dots,E_m$ be mutually exclusive events where $E_j$ represents the event that pool $j$ successfully issues the block.  Miners can choose to join one or several mining pools, and their initial risk vector is given by $\XX=P(\vmathbb{1}_{D_1},\dots,\vmathbb{1}_{D_n})$ as in Section \ref{sec:71}.
Since the $n$ miners form the $m$ mining pools, we have the equality
  $\bigcup_{i=1}^n D_i =  \bigcup_{j=1}^m E_j$ and the decomposition $\vmathbb{1}_{D_i}= \sum_{j=1}^m \vmathbb{1}_{D_{i}\cap E_j}$ where $D_{i}\cap E_j$ is the contribution of miner $i$ to pool $j$.  
  
Due to the separation of mining pools, we consider a generalized reward sharing rule 
with target information $\mathcal G=\sigma (P,\vmathbb{1}_{E_1},\dots,\vmathbb{1}_{E_m})$, that is, the information about the Bitcoin price $P$ and the winning pool which successfully mines the block.
For this choice of $\mathcal G$, the generalized CMRS $ \widehat {\mathbf A}$ in \eqref{eq:generalized-def} is given by
\begin{equation}
\label{eq:multiplepool}
  \widehat A_i^{\XX|\mathcal G} = \sum_{j=1}^m \frac{\p(D_{i}\cap E_j)}{\p(E_j)} P \vmathbb{1}_{E_j}.
\end{equation}
Note that this rule can be easily implemented in practice as $ \p(D_{i}\cap E_j)/ \p(E_j)$ is the relative share of computational contribution of miner $i$ to  pool $j$. 
This rule can be equivalently explained by the mechanism in which all $m$ pools are allocated separately and each of them uses CMRS.

Similarly to Proposition \ref{prop:bitcoin}, we can show that \eqref{eq:multiplepool} is the only generalized reward sharing rule with target information $\mathcal G$,
and this rationalizes the practice of allocating rewards across multiple mining pools.

\subsubsection{Multiple Blockchains} 
\label{subsec:73}

We proceed to consider a pool of $n$ miners with a collection of $m$  cryptocurrencies (which we call coins) with random prices $P_1,\dots,P_m$ in a pre-specified period of time.
The computational contribution of miner $i$ to coin $j$ is fixed during this period of time.
For simplicity, we assume that for each of these coins, at most one block may be issued during this period of time. 
Denote by $D_{ij}$   the event that miner $i$   issues the block for coin $j$,
and by $D^j=\bigcup_{i=1}^n D_{ij}$ the event that coin $j$ is successfully mined by the pool. 
The events $D_{ij}$ are mutually exclusive across $i\in [n]$ for the same $j $. We further assume that  $D_{ij}$ is independent of $\{ D_{k\ell}: k\in [n],~\ell \in [m]\setminus\{ j\}\}$,  because issuing the block of one coin should not affect issuing the block of another one. 

Similarly to Section \ref{subsec:71}, we assume that the prices $P_1,\dots,P_m$ are independent of the issuance events. 
  The initial risk vector is given by $\XX=\sum_{j=1}^m P_j(\vmathbb{1}_{D_{1j}},\dots,\vmathbb{1}_{D_{nj}})$.
  Finally, we consider the target information 
 $\mathcal G=\sigma (P_1,\dots,P_m,\vmathbb{1}_{D^1},\dots,\vmathbb{1}_{D^m})$,
 which is information about the coin prices and the events of whether each of them is successfully mined.  
 For this choice of $\mathcal G$, the generalized CMRS $  \widehat {\mathbf A}$ in \eqref{eq:generalized-def} is given by
 $$ 
  \widehat A_i^{\XX|\mathcal G} =\sum_{j=1}^m \frac{\p(D_{ij})}{\p(D^j)} P_j \vmathbb{1}_{D^j} ,
 $$
because $\E[ P_j\vmathbb{1}_{D_{ij}}|\mathcal G] =  \E[ P_j\vmathbb{1}_{D_{ij}} | P_j \vmathbb{1}_{D^j} ] = P_j \vmathbb{1}_{D^j} \p(D_{ij})/\p(D^j)  .$
In other words, each miner gets a proportion $\p(D_{ij}) /\p(D^j)$ of each successfully mined coin, where the proportion is determined by its relative contribution to the pool for that particular coin.

\subsection{Revenue Sharing of Digital Content}
\label{sec:74}

Our next application concerns revenue sharing in subscription‑based online platforms
for digital content. The primary examples are   music platforms such as Spotify, Deezer, or Apple Music; see \cite{MKAC23} for a description of different revenue sharing schemes in subscription‑based music platforms.

Different remuneration rules lead to different revenue sharing outcomes. Suppose we have 10 users, with one user paying a \$100 subscription fee to the platform to listen to Musician A, while the remaining nine users each pay a \$1 subscription fee to listen to Musician B. Musician A would get most of the sharing revenue under the user-centric rule, but Musician B would get most of the sharing revenue under the pro-rata rule as Musician B gets most hits. Most platforms use the pro-rata rule instead of the user-centric rule; we refer to \cite{MKAC23} and \cite{BM25} for a comparison of different revenue sharing rules. 

Suppose that there are $n$ artists  and $m$ potential users.
Consider a total of $T$ time periods. For example, $T=60\times 24$ if we divide each day into one-minute units.
The number $m$ is usually much larger than $n$.
Let
 $D_{ij,t}$ be the event that user $j\in [m]$ listens to  the music of artist $i\in [n]$ on the platform   
  at time $t\in[T]$.
 We assume that the events $D_{ij,t}$ are mutually disjoint across $i\in [n]$. 
In other words, at any given moment $t$, user $j$ can only listen to the music of one artist.  

Assume that the subscription events across different users are independent at any given moment; i.e., $D_{ij,t}$ is independent of $\{ D_{k\ell,t}: k\in [n],~\ell \in [m]\setminus\{ j\}\}$ for each $i\in [n]$ and $j\in [m]$. 
Let $D_t^j=\bigcup_{i=1}^n D_{ij,t}$ be the event that at time $t$ user $j$ is using the platform to listen to music. 
The usage of the platform generates a non-random revenue $q_{j,t}>0$ from user $j$ (i.e., incremental subscription fee paid at time $t$, which may vary across users). 
If $D_t^j$ does not occur, i.e., user $j$ does not listen to any music at time $t$,
we assume that either user $j$ does not pay for the fee at time $t$ or the fee paid at that time does not go to any artist. 
Suppose that for $j \in [m]$, a proportion $\delta_{j,t}$ of paid subscription fees will be shared by
the artists, while the rest is kept by the platform.
The initial risk vector at time $t$  is given by $\XX_t=\sum_{j=1}^m p_{j,t}(\vmathbb{1}_{D_{1j,t}},\dots,\vmathbb{1}_{D_{nj,t}})$, where
$p_{j,t} = \delta_{j,t} q_{j,t}$. 
The total initial risk vector over the period $[T]$ is $
\mathbf X=\sum_{t=1}^T \mathbf X_t$. 
The target information is modeled by $\mathcal G=\sigma (\vmathbb{1}_{D_t^1},\dots,\vmathbb{1}_{D_t^m}, t \in [T])$, that is, the information based on the events of subscription. 
We assume that, conditional on $D_t^j$, the artist choice of user $j$ at time
$t$ is independent of the remaining activity indicators in $\mathcal G$. 
For this choice of $\mathcal G$, the generalized CMRS in \eqref{eq:generalized-def} gives
$$
\widehat A_i^{\mathbf X|\mathcal G}
=
\sum_{t=1}^T
\sum_{j=1}^m
\frac{\mathbb P(D_{ij,t})}{\mathbb P(D_t^j)}
p_{j,t}\vmathbb{1}_{D_t^j},
\qquad i\in[n],
$$
similarly to the model in Section \ref{subsec:73}.

For simplicity, let us assume that the probability of user $j$ listening to a particular artist $i$ 
is stationary across the period $[T]$; this is quite realistic if the period is not long. That is, we can denote by $d_{ij}=\p(D_{ij,t})$ and $d^{(j)}=\p(D_t^j).$
Hence, 
\begin{equation}
\label{eq:revenue}
   \widehat {A}_{i}^{\XX|\mathcal G} =
   \sum_{j=1}^m 
   \frac{d_{ij}}{d^{(j)}}
  \sum_{t=1}^T   p_{j,t} 
   \vmathbb{1}_{D_t^j}
   = \sum_{j=1}^m 
   \frac{d_{ij}}{d^{(j)}}p_j, 
\end{equation} 
where $p_j =  \sum_{t=1}^T   p_{j,t} 
   \vmathbb{1}_{D_t^j} $
   is the total fee collected from user $j$ to be shared with artists, observable to the platform and easy to verify for accounting purposes.

Although the probabilities $d_{ij}$ and $d^{(j)}$ are not directly observable, their ratio $d_{ij}/d^{(j)}$ can be estimated by $s_{ij}/s^{(j)}$, where $s_{ij}$
is the number of streams by user $j$ of artist $i$'s work, and $s^{(j)}$ is the total number of streams by user $j$ over the period $[T]$. 
With the ratio $d_{ij} / d^{(j)}$ estimated by $s_{ij}/s^{(j)}$, the revenue sharing mechanism \eqref{eq:revenue} is exactly the \emph{user-centric}  remuneration rule promoted by some platforms based on an argument of fairness.\footnote{For instance, the platform \texttt{Deezer} is promoting the {user-centric} payment system; see \url{https://www.deezer-blog.com/how-much-does-deezer-pay-artists/} (accessed April 2023).} 
Thus, 
our framework provides a theoretical justification  for the user-centric rule of sharing digital content.

\subsection{The Tontine P2P Scheme as a Way to Hedge Longevity Risk}
\label{sec:81new}

Traditional annuity products, which promise predictable guaranteed lifetime income for policyholders, face stringent regulatory capital reserve requirements to meet the conservative estimation of longevity risk. The tontine scheme, promoted by Lorenzo de Tonti in the 1650s, after whom the scheme is named, has recently received renewed attention in both the academic and industrial literature; see \cite{Sabin2010}, \cite{MS2015new} and \cite{CHK19}. 
By pooling, a tontine hedges the longevity risk via internal risk transfer rather than via capital reserves as in annuity contracts.  
A typical tontine provides significant cash flows only for the shrinking group of survivors among initial participants. Thus, the participants who live longer receive larger payments.       

We focus on a particular form of tontine scheme, which preserves Axiom AF. Consider a one-period model in which the total non-random return of the tontine fund is $R$, and individual $i$ contributes $c_{0,i}$ at time 0, $i=1,2,...,n.$ Thus, the $i$-th individual who survives at time 1 should get payment $c_i=c_{0,i}R$ plus an additional payment known as a mortality credit. More precisely, the total payment to the deceased participants, 
$S=\sum_{j=1}^n X_j$, 
should be redistributed among all the surviving participants and the beneficiaries of the deceased, yielding the mortality credit, where $X_j=(1-I_j)c_j$, 
$I_j=1$ if individual $j$ survives at time 1, and $I_j=0$ otherwise. We assume independence among mortality indicators. Assume also that at least one participant survives the period; otherwise there are no participants to share the fund.\footnote{To maintain actuarial fairness, one solution in the rare situation of no participants surviving is to give back $c_i$ to each participant $i$'s beneficiary; see e.g., \cite{DY2017new} and \cite{BFV2024}.} There are at least two different sharing rules that can be applied to share the total amount $S$. The first one as in \cite{DY2017new} uses the MPRS rule, yielding 
\begin{equation}
    {\bar Y}_i=\frac{\E [X_i]}{\E [S]} S=
\frac{c_iq_i}{\sum_{j=1}^n c_jq_j} S,
\label{e:new8.3.1}
\end{equation}
where $q_j=\p (I_j=0)$ is the death probability for the $j$-th individual. More precisely, the $i$-th individual, if she survives, gets a total payment of $c_i+{\bar Y}_i.$ The second one as in \cite{D19} relies on the CMRS, leading to 
$Y_i=\E [X_i |S]$ and a total payment of 
$c_i+{Y}_i.$  Using the size-biased distribution,   \cite{D19} computes $Y_i$ for a discrete distribution explicitly via the representation
\begin{equation}
    Y_i = \E [X_i |S=s]=c_i \frac{q_i \p (\sum_{j\neq i } X_j =s-c_i)}{\p (S=s)}.
\label{e:new8.3.2}
\end{equation}
 Note that MPRS does not satisfy Axiom RF and Properties UI and CP; thus, our axiomatic approach favors the use of the CMRS sharing rule.\footnote{On the other hand, some optimality of MPRS in a different tontine model is shown by \cite{BFV2024}.}

\subsection{Claim-Only P2P Insurance}
  A claim-only risk sharing scheme in P2P insurance is designed such that the participants who incurred a loss will share the total loss, and the participants who do not incur a loss will not need to pay, an intuitively compelling feature for P2P insurance in a large pool; see \citet[Section 7.3]{DDR22}. 
  For instance, 
 the claim-only CMRS in \cite{DDR22} is defined as 
$
 A_i^{\mathbf X} = \E\left[X_i|S^{\mathbf X},\vmathbb{1}_{\{X_1>0\}},\dots,\vmathbb{1}_{\{X_n>0\}}\right],$
 $i\in [n],
$
where $\mathbf{X} \in (L_+^1)^n$ are nonnegative losses. Hence, we get 
$$
\mathbf A^{\mathbf X} = \E\left[\mathbf X|S^{\mathbf X},\vmathbb{1}_{\{X_1>0\}},\dots,\vmathbb{1}_{\{X_n>0\}}\right] = \E\left[\mathbf X|\mathcal G^{\mathbf X}\right],
$$
where  $\mathcal G=\sigma\left(\vmathbb{1}_{\{X_1>0\}},\dots, \vmathbb{1}_{\{X_n>0\}}\right)$ in Section \ref{sec:generalized}.  
This risk sharing rule is CMRS with target information. 

\subsection{P2P Health Care Insurance: A Case Study of Xiang Hu Bao}
\label{sec:73}

This case study aims to discuss the choice of the probability  and expectation used in Axiom AF.
Our previous theoretical results only indicate that the same expectation operator $\mathbb{E}$
should be used in Theorems \ref{th:full} and \ref{th:full3} and in Axiom AF. But
how should one choose the expectation operator $\mathbb{E}$ in a practical risk sharing context, with possibly dynamically varying membership? 

There are at least three choices. First, the population expectation may be taken. Second, the sub-population expectation among the risk sharing members may be taken.
Third, a subjective expectation may be taken. The third one is
relevant if there is an endogenous effect due to the announcement of the
risk sharing rule; for example, the risk sharing rule may create demand
induction by attracting certain groups of participants to the risk sharing
pool in the future. This case study is based on Xiang Hu Bao (XHB), a popular
online mutual risk-sharing platform in China, shedding some light on the choice of
the expectation $\mathbb{E}$.\footnote{Formally speaking, XHB is not a decentralized risk sharing scheme, since Ant Financial serves as a central party and makes payments to participants.
The claim payment per member is calculated twice a month by adding 8\% platform operating expenses to the total claim and then dividing the total claim equally among the pool members. The 8\% operating expenses are relatively low partly because XHB uses blockchain technology. We ignore the 8\% charge when discussing risk sharing rules.
}

Launched in 2018, XHB was operated by the Chinese conglomerate Ant Financial as
a novel FinTech platform in addition to the existing insurance
venues in China, such as traditional health insurance and Critical Illness
Insurance (CII).  
Within its first year of operation, XHB achieved a membership base exceeding 100 million, comparable to the total size of the CII market in China.
However, without obtaining regulatory approval, XHB ceased operations on January 28, 2022. The risk sharing platform XHB satisfies our definition of a risk sharing rule because the total money is reallocated among all participants (including members and Ant Financial). Nevertheless, we do not assume that the risk sharing rule is CMRS. (For instance, Axiom OA cannot be verified in this context.)

To understand the risk sharing rule used by XHB, we first review the eligibility of the XHB sharing. First, the applicants or their immediate family members must be members of Alipay (the Chinese
counterpart of Paypal). Nearly 40\% of the Alipay members
are 
under 25 years of age, 
and only 9\% are over 40 years of age, 
while nearly half of the Chinese population is over 40. Thus, the applicants tend to be much younger than the Chinese
population. Second, XHB\ requires a minimum credit score of 600 points out
of the maximum of 950 points. Since wealthier people tend to be healthier, the
applicants of XHB also tend to be healthier. Third, the applicants must meet
some pre-defined medical requirements, such as not having been hospitalized
for over 15 days during the last two years. Once approved, there is also a 90-day probation period in which no indemnity payment would be made. 

The indemnity payment is much lower than that of typical
medical insurance. Basically, there is one main risk sharing pool in XHB: the
critical illness plan (with a list of about 100 types of critical illness)
which provides one-time indemnity 300,000 (Chinese Yuan) CNY (1 USD is about
7 CNY) to members 
not above 39 years of age 
 and 100,000 CNY to members 
 aged 40 to 59 years. 
 The critical illness plan was selected by over 90\% of the participants.\footnote{Although
the critical illness plan is the most popular one, there is also a senior cancer plan to cover cancer-related indemnity
100,000 CNY for members 60 to 70 years old. Other minor plans, such as the mild critical illness plan and reduced payment
plans, were also briefly introduced during short periods.}  See
\cite{Fang2023} and \cite{Feng2023} for details. 

As emphasized in a recent study by \cite{Fang2023}, an important consequence of the low
indemnity payment is that low-risk people are self-selected to XHB, while
high-risk people are self-selected for traditional medical insurance and CII
programs. Indeed, the aggregate claim payment in the month of Jan 2019 was
only 0.03 CNY per member, while in the more recent month of Dec 2020 was
10.14 CNY per member. Thus, the population 
expectation 
 for 1.4 billion
people and the subpopulation 
expectation 
 for 100 million members are quite
different.

Since XHB treats all members not above 39 years of age identically, and all members aged 40 to 59 years identically, 
the total random payment to be shared can be written 
$$
S= \sum_{i=1}^{n_1} X_i^{(1)} + \sum_{j=1}^{n_2} X_j^{(2)},~~~~X_i^{(1)}=(1-I_i^{(1)})c^{(1)},~~~
X_j^{(2)}=(1-I_j^{(2)})c^{(2)},$$
where $X_i^{(1)}$ and 
$X_j^{(2)}$ 
 are the payments to the $i$-th individual member  
 in the first group (not above 39 years of age)  and the $j$-th member in the second group (aged 40 to 59), 
 respectively. The one-time indemnity 
$c^{(1)}=300,000$ and $c^{(2)}=100,000$, and $I_i^{(1)}=1$ if the individual 
not above 39 years of age 
 does not have a critical illness, and $I_i^{(1)}=0$ otherwise. 
 $I_j^{(2)}$ 
 is defined similarly.   

If the MPRS rule is used, then similar to \eqref{e:new8.3.1}, the premium charge should be 
$$ 
{\bar Y}_l^{(k)}=\frac{\E [X_l^{(k)}]}{\E [S]} S=
\frac{c^{(k)}q_l^{(k)}}{\sum_{i=1}^{n_1} c^{(1)}q_i^{(1)}+\sum_{j=1}^{n_2} c^{(2)}q_j^{(2)} } S,
~~~k=1,2,$$
where 
$
q_i^{(1)}=\p(I_i^{(1)}=0),$ and $q_j^{(2)}=\p(I_j^{(2)}=0).
$
Since the premium charge is identical for two groups, we must have
$c^{(1)}q_i^{(1)}=c^{(2)}q_j^{(2)},$ i.e., under the MPRS rule, the implied ratio of incidence rates must be
$$\frac{\p(I_j^{(2)}=0)}{\p(I_i^{(1)}=0)}
=\frac{300,000}{100,000}=3.$$
The calculation based on \eqref{e:new8.3.2} for the CMRS rule is much more involved. However, a simplification is possible if $n_1,~n_2 \to \infty$, quite reasonable with over 100 million members. In this case the reverse martingale
$\E[X|S]$ goes to $\E [X]$ under suitable regularity conditions; see \cite{DR21a}. 
{
Thus, the premiums for the two groups of members are
$c^{(k)}q_l^{(k)}$, $k=1,2$, if $n_1,~n_2 \to \infty$. Again, we have 
\begin{equation}
\label{e:new8.3.3}  c^{(1)}q_i^{(1)}=c^{(2)}q_j^{(2)},~~~\frac{\p(I_j^{(2)}=0)}{\p(I_i^{(1)}=0)}=3. 
\end{equation}
}
If the uniform risk sharing rule is applied along with the actuarial fairness axiom, then 
the fact that the same premium rate is used for both groups also yields \eqref{e:new8.3.3}. Thus, all three risk sharing rules point to the same implied ratio of incidence rates.


\begin{table}[!htbp]
\caption{Ratios of incidence rates and related ratios. CII means the participants in the national Critical Illness
Insurance in China.}\label{tab:XHB}
\begin{center} 
\begin{tabular}{cccc} 
XHB Implied Ratio & XHB Actual Ratio & CII Ratio & Chinese Population Ratio \\ \hline 
3 & 4.31 & 5.12 & $\allowbreak $9.92 \\ 
& ($\mathrm{SD} = 0.21$) & & (in 2013) \\\hline 
\end{tabular}
\end{center} 
\end{table}

Table \ref{tab:XHB}   compares this implied
ratio of incidence rates and related ratios.
The second and the third columns are essentially converted  from the t-statistics reported in Table A1 in
\cite{Fang2023}, while the last column is obtained from Table
8.4 in \cite{Feng2023} by using the morbidity table published by China Insurance Regulatory
Commission in 2013. Under population 
expectation, 
the incidence ratio should be around 9.92. For the subpopulation of XHB, the ratio is estimated to be around 4.31, while the ratio is about 5.12 for a related subpopulation CII. However, under the subjective probability 
 implied by XHB, the ratio is only 3.

One possible explanation for such big differences among the ratios is that
XHB uses subjective expectation 
 by taking into consideration
that the very favorable low cost for members 
 aged 40 to 59 
 may
induce future demand among this group of people. 
Attracting the more senior, and perhaps wealthier, cohort may benefit Alipay through mandatory Alipay membership.
Indeed, the allocated cost per member steadily increased from 0.03 CNY per month in January 2019 to 10.14 CNY in December 2020. Furthermore, low-risk younger customers may leave XHB, as they might perceive the incidence ratio as unfavorable. 

The example of XHB sheds some light on 
how to choose the expectation $\mathbb{E}$ in practice. Although XHB only existed for fewer than four years due to the lack of proper regulatory licenses, it is an interesting social experiment on anonymized risk sharing involving over 100 million people.

\section{Conclusion}
\label{sec:conclusion}

Anonymity guarantees that agents are not distinguished by their preferences, identities, private operations, or realized losses. 
Anonymized risk sharing is especially relevant in the digital economy, with applications including blockchain mining pools, revenue sharing in digital arts, and P2P health insurance. 
Although there is extensive literature on axiomatic approaches in decision theory, so far there is no axiomatic theory for risk sharing.  
This paper presents 
an axiomatic theory of anonymized risk sharing from a perspective of normative economics.
Based on two axioms of fairness and two axioms of anonymity, we prove that the conditional mean risk sharing rule is the only risk sharing rule that satisfies the four axioms.
We do not see any general reasons to dispute any of the four axioms in the framework of anonymized risk sharing as a whole, although in certain specific applications the conditional mean risk sharing rule may be suboptimal for some individual agents. 
Therefore,  the conditional mean risk sharing rule holds normative appeal and is well-suited for practical applications   in decentralized risk sharing.

\clearpage
\newpage

\normalsize
\appendix

\pagenumbering{arabic}
\renewcommand*{\thepage}{O-\arabic{page}}
\setcounter{page}{1}

\begin{center}
{\bf \Large Online Supplement}

{\Large 
Anonymized Risk Sharing}

\end{center}

\section{Proof sketch of Theorem \ref{th:full} in the case of a discrete space}
\label{sec:AppA}

For a discrete $\Omega$, the main idea is to analyze each possible realized value $s\in \R$ of $S^\XX$ one by one.  There are at most countably many such $s$.  
  Let $\mathbf A$ be a risk sharing rule satisfying the four axioms. We focus on the allocation to agent $1$,
  and aim to show $A_1^\XX=\E[X_1|S^\XX]$ for all $\XX\in \X^n$; the allocations to the other agents are similar. Fix  $S \in \X$, and we will first consider the risk vector $(X,S-X,0,\dots,0)$ by allowing $X$ to vary within $\X$. Let $V$ be the set of possible values taken by $S$.
     By RA, the value of the allocation $A_1^{(X,S-X,0,\dots,0)}$ to agent $1$ is determined by the realized value $s \in V$ of $S$ and $X$. Denote this value  by $h^{S,s}(X)$, that is,  for   fixed $S \in \X$ and $s \in V$, 
      \begin{align}
      \label{eq:h}
      h^{S,s}(X) = A_1^{(X,S-X,0,\dots,0)} \mbox{ given   $S=s$}.
      \end{align} 
	We can carefully check that $h^{S,s}:\X\to \R$  satisfies the following properties (in the last property we allow $s$ to vary in $V$):
	\begin{enumerate}[(a)]
		\item normalization: $h^{S,s}(t) = t$ for all $t\in \R$; \hfill (by \eqref{eq:constancy}) 
		\item additivity: $h^{S,s} (X+Y) = h^{S,s} (X) + h^{S,s}(Y) $ for $X,Y\in \X$; \hfill (by OA and (a)) 
		\item monotonicity: $h^{S,s}(Y) \ge h^{S,s}(X)$ if $Y\ge X$; \hfill (by RF and (b)) 
		\item $h^{S,s}(S) =s$; \hfill (by RF and \eqref{eq:trivial})  
		 \item[(e)] $\sum_{t\in V} h^{S,t}(X) \p(S=t) =\E[X]$ for $X\in \X$. \hfill (by AF)  
	\end{enumerate}
	The properties (a), (b) and (c) together guarantee that $h^{S, s}$ is linear, monotone, and normalized. Using a standard representation theorem (such as that of Riesz), there exists   a probability measure $P_{S,s}$  such that  
	\begin{equation}\label{eq:rep}
		h^{S,s} (X) = \int  X \dd P_{S,s}  ~~~\mbox{for all $X\in \X$}.
	\end{equation} 
The next task is to show that $P_{S,s}$ is precisely the conditional probability $\p(\cdot |S=s)$.  
Let $x\wedge y$ represent  the minimum of $x,y\in \R$.
Using (d) and taking $X = S \wedge x$ in \eqref{eq:rep}, we  arrive at  
	 $$
	h^{S,s} (S\wedge x)  \le    \int S \dd P_{S,s} \wedge \int x \dd P_{S,s}  =s\wedge x.
	 $$ 
	This inequality  further implies 
 $$
	\E[S\wedge x] = \sum_{t \in V}  (t\wedge x) \p(S=t)  \ge  \sum_{t \in V}   h^{S, t} (S\wedge x)   \p(S=t)   = \E[S\wedge x],
 $$ 
 where the last equality is due to (e). 
	Hence, we get
	$
	h^{S,s} (S\wedge x) =s\wedge x 
	$   
	for each $s\in V$, and this gives, in particular,  
	$
	\int (S\wedge s)  \dd P_{S,s}  = s.
	$ 
	Therefore,  $P_{S,s} ( S\ge s )=1$. Using symmetric arguments, we can show   $P_{S,s} ( S\le s )=1$. As a result, 
 $
	P_{S,s} ( S = s )=1.
 $
 Using this  equality and (e), for any $B\subseteq \{S=s\}$,   we have
   $$\p(B|S=s)   =
\frac{   \E[\id_B]  } {\p(S=s)} =  \frac{\sum_{t\in V}  h^{S,t}(\id_B) \p(S=t ) } {\p(S=s)}  = \frac{  \sum_{t\in V}  P_{S,t}(B) \p(S=t)  } {\p(S=s)}=P_{S,s}(B) .
  $$
Therefore,
	$ 	P_{S,s} (\cdot) = \p(\cdot |S=s)
	$ 	and
	$ 
	h^{S,s}(X) =  \E[X|S=s] $ for $X\in \X$.
Based on this result, we can finally get $A_1^\XX =\E[X_1|S^\XX]$ for a general $\XX$ using $
A_1^{\XX} = A_1^{(X_1,S^\XX-X_1,0,\dots,0)}
$ guaranteed by OA.
This concludes the proof of Theorem \ref{th:full} in case $\Omega$ is discrete.

\section{Proofs in Section \ref{sec:axiom}}

\label{app:pf4}

We first prove Theorem \ref{thm:BT}, as it will be used in the proof of Theorem \ref{th:full}.

\begin{proof}[Proof of Theorem \ref{thm:BT}]
	First, we prove that $\phi$ is continuous. Suppose that $X_n\to X$ in $L^1(\Omega,\mathcal F, \p)$.   By using (b) and (c), we have 
	$$
	\phi(X_n)-\phi(X) =\phi(X_n-X)\le \phi(|X_n-X|).
	$$
	Similarly,
	$$
	\phi(X_n)-\phi(X) =\phi(X_n-X) \ge \phi(- |X_n-X|) = -\phi(  |X_n-X|).
	$$
	Hence, 
	$$
	|\phi(X_n)-\phi(X) |\le | \phi(|X_n-X|)|=\phi(|X_n-X|),
	$$
	where the last equality is due to $\phi(|X_n-X|) \ge \phi(0)=0$ by (a). 
	Using (e), we have $\E[\phi(|X_n-X|)]=\E[|X_n-X|] \to 0$.
	Therefore, $\E[|\phi(X_n)-\phi(X) |]\le \E[\phi(|X_n-X|)]\to 0$. This means $\phi(X_n)\to \phi(X)$ in $L^1(\Omega,\mathcal G, \p)$, thus showing the continuity of $\phi:L^1(\Omega,\mathcal F, \p) \to L^1(\Omega,\mathcal G, \p)$.
	
	Next, we prove that $\phi$ is linear. Based on the fact that $\phi$ is (b) additive and (c) monotone, we have 
	\begin{equation*}\label{eq:positive2}
		\phi(X) \ge \phi(0) = 0 ~~~~ \text{ for any }  X \in L^1(\Omega, \mathcal F, \p) \text{ and } X \ge 0,
	\end{equation*}
	which implies that $\phi$ is linear (see, e.g., Theorem 1.10 of \cite{AB06}):
	\begin{align*}
		\phi(\alpha X + \beta Y) = \alpha \phi(X) + \beta \phi(Y) ~~~~ &\text{ for any } \alpha, \beta \in \R \\ &\text{ and } X, Y \in L^1(\Omega, \mathcal F, \p).
	\end{align*}
	It further follows that $\phi$ is a positive operator on $\X$.  Recall that a linear operator between two ordered vector spaces is  a positive operator  if it maps positive elements to positive elements. 
	
	Finally, we show that $\phi$ satisfies  the following      property 
 	\begin{equation}\label{eq:SBT}
 		\text{$\phi(X)=X$ for all $\sigma(S)$-measurable $X$.}
 	\end{equation} For $t\in \R$,  by (c) we have 
	$
	\phi(S\vee t) \ge \phi(S)\vee \phi(t)$,
	and by (a) and (d) we get
	$ \phi(S)\vee \phi(t)= S\vee t.
	$ 
	Since $\phi(S\vee t) \ge S\vee t$ and they have the same mean by (e),   we know 
	\begin{align} \label{eq:upper}
		\phi(S\vee t)=S\vee t ~~~~\mbox{for all $ t\in \R$.}
	\end{align}
	Write $T_{s,t} =  \frac{1}{t-s}  (S\vee t - S \vee s ) $ for $t>s$. 
	Note that for all $t>s$, 
	\begin{align} \label{eq:sandwich} \id_{\{S\le s\}} \le  T_{s,t}   \le \id_{\{S\le t\}}.\end{align}
	The linearity of $\phi$ and  \eqref{eq:upper} imply that 
	$\phi\left( T_{s,t} \right) =  T_{s,t}$.
	Using the above equality, (c) and 
	\eqref{eq:sandwich}, 
	we have, for all $t>s$,
	$$
	\phi\left(  \id_{\{S\le s\}} \right) \le
	\phi\left( T_{s,t} \right) =  T_{s,t} \le \id_{\{S\le t\}},
	$$ 
	and 
	$$
	\phi\left(  \id_{\{S\le t\}} \right) \ge
	\phi\left( T_{s,t} \right) =  T_{s,t} \ge \id_{\{S\le s\}}.
	$$ 
	It follows that for all $\epsilon>0$, 
	$$
	\id_{\{S\le t-\epsilon \}} \le \phi\left(  \id_{\{S\le t\}} \right) \le \id_{\{S\le t+\epsilon \}} .
	$$
	Hence, $\id_{\{S< t\}} \le \phi\left(  \id_{\{S\le t\}} \right) \le \id_{\{S\le t\}}   $
	for all $t\in \R$.
	Using (e), we know  
	$ \phi\left(  \id_{\{S\le t\}} \right)  = \id_{\{S\le t\}}$.   
	From this equality, using (a) and linearity of $\phi$, it follows that $\phi\left(  \id_{\{S> t\}} \right) = \id_{\{S> t\}} $ for all $t\in \R$. 
	
	Define the class
	$$
	\mathcal{C} = \left\{C \in \mathcal{F}: \phi(\id_{C}) = \id_{C}  \right\}.
	$$
	Hence, $\{S \leq t\} \in \mathcal{C}$ for any $t \in \R$. We have 
	$\Omega \in \mathcal{C}$. Using linearity of $\phi$, we have that if $C \in \mathcal{C}$, then 
	$$
	\phi(\id_{C^c}) = \phi(1 - \id_{C}) = \phi(1) - \phi(\id_{C}) = 1 - \id_{C} = \id_{C^c},
	$$
	which implies the complement set $C^c \in \mathcal{C}$. 
	Suppose that $\{C_i\}_{i \geq 1} \subseteq \mathcal{C}$ are disjoint. We proceed to show that $\bigcup_{i = 1}^ \infty C_i \in \mathcal{C}$. Indeed, using monotonicity and additivity of $\phi$, we have
	$$
	\phi\left(\id_{\{\bigcup_{i = 1}^{\infty} C_i\}}\right) \geq \phi\left(\id_{\{\bigcup_{i = 1}^{m} C_i\}}\right) = \phi\left(\sum_{i = 1}^{m} \id_{C_i}\right) = \sum_{i = 1}^{m} \id_{C_i}, ~~ \text{for all } m \geq 1.
	$$
	Letting $m \rightarrow \infty$, we have
	$$
	\phi\left(\id_{\{\bigcup_{i = 1}^{\infty} C_i\}}\right) \geq \sum_{i = 1}^{\infty} \id_{C_i} = \id_{\{\bigcup_{i = 1}^{\infty} C_i\}}.
	$$
	Based on (e), we have
	$$
	\E\left[\phi\left(\id_{\{\bigcup_{i = 1}^{\infty} C_i\}}\right)\right] = \E\left[\id_{\{\bigcup_{i = 1}^{\infty} C_i\}}\right],
	$$
	which implies   $\phi\left(\id_{\{\bigcup_{i = 1}^{\infty} C_i\}}\right) = \id_{\{\bigcup_{i = 1}^{\infty} C_i\}}$ and $\bigcup_{i = 1}^{\infty} C_i \in \mathcal{C}$. 
	Hence, the class $\mathcal{C}$ is a $\sigma$-field and $\sigma(S) \subseteq \mathcal{C}$ based on the monotone class theorem. 
	It follows that  
	$\phi(\id_{B}) = \id_B$ for all $B\in \mathcal G$. 
	Since any $\mathcal G$-measurable $X$ can be upper and lower approximated by the summation of simple functions, using linearity and monotonicity we conclude that $\phi(X)=X$ for all $\mathcal G$-measurable $X$.
	
The conditions  that $\phi$ is continuous, linear and monotone and satisfies  \eqref{eq:SBT} guarantee   the representation of $\phi$    (see Proposition 2.6 of \cite{FKV12} or Theorem 1 of \cite{P67}), as 
	\begin{equation}\label{eq:represent}
		\phi (X) = \E[Z X | S] \mbox{~~~~~for all $X\in L^1(\Omega,\mathcal F,\p)$,}
	\end{equation}
	for some $Z \ge 0$ satisfying  $\E[Z | S] = 1$. 
	Using (e), we get
	 $
	1= \E[Z] =\E[\phi(Z)] = \E[\E[Z^2 |S]] = \E[Z^2].
	 $ 
 Since $\E[Z^2]=\E[Z]=1$, we know $Z = 1$.
	Hence, we have
	$ 
	\phi(X) = \E[X | S] $ for $X \in \X 
	$ and
	this completes the proof.  
\end{proof}

\begin{remark}\label{rem:Lq}
The key step in the proof of Theorem \ref{thm:BT} is to obtain the property  \eqref{eq:SBT}; $\phi$ with such a property is sometimes called a projection. Several characterizations of  the conditional expectation directly rely on this property; see \cite{P67}   and \cite{EFHN15}. 
In particular, Theorem 1 of \cite{P67} holds for subspaces of $L^1(\Omega,\mathcal F,\p)$, and hence our result  in Theorem \ref{th:full} holds for general $\X=L^q$ where $q\in [1,\infty]$. 
\end{remark}

\begin{proof}[Proof of Theorem \ref{th:full}]

	The ``if" statement  is checked in Section \ref{sec:axiom}. To show the   the ``only if" statement, we separate the two cases of $\X=L^1$ and $\X=L^1_+$. 
	
	\noindent (i)  \underline{The case   $\X=L^1$.}
	
	  Let $\mathbf A$ be a risk sharing rule satisfying Axioms AF, RF, RA and OA.  
	Fix any $S \in \X$. 
	Define the mapping
	$$
	h^{S}:\X\to L^1(\Omega, \sigma(S), \p), ~ X\mapsto A_1^{(X,S-X,0,\dots,0)}.
	$$
	Note that RA guarantees that $h^{S}$ takes values  in  $L^1(\Omega, \sigma(S), \p)$. 
	We will verify that $h^{S}$  satisfies the following properties on $\X$:
	\begin{enumerate}[(a)]
		\item constant preserving: $h^{S}(t)=t$ for all $t\in \R$; 
		\item additivity: $h^{S} (X+Y) = h^{S} (X) + h^{S}(Y) $ for $X,Y\in \X$; 
		\item monotonicity: $h^{S}(Y) \ge h^{S}(X)$ if $Y\ge X$;
		\item $h^{S}(S) = S$;
		\item $\E[h^S(X)] = \E[X]$ for $X \in \X$;
	\end{enumerate}
	First, (a) follows directly from \eqref{eq:constancy} and the definition of $h^{S}$. Next, 
	we proceed to prove (b). By using Axiom OA, we have, for any $X,Y\in \X$ (note that here we use the fact that $n \ge 3$),
	\begin{align}
		h^{S} (X+Y)& = A_1^{(X+Y, S-X-Y, 0, \dots, 0)} \notag
		\\& =  A_1^{(X,S-X-Y,Y,0,\dots,0)} + A_3^{(X,S-X-Y,Y,0,\dots,0)} \notag
		\\& =  A_1^{(X,S-X,0,\dots,0)} + A_3^{(0,S-Y,Y,0,\dots,0)} \notag
        \\&  =  h^{S}(X) + A_3^{(0,S-Y,Y,0,\dots,0)} \label{eq:rw1}
	\end{align}
	where Axiom OA is used in the second and third equalities. 
	In particular, by choosing $X=0$ and using the fact that $h^{S}(0) = 0$ in (a), \eqref{eq:rw1} implies
	 $
	h^{S} (Y)= A_3^{(0,S-Y,Y,0,\dots,0)}.
	 $     
	Using this relationship and \eqref{eq:rw1}, we further have 
	$$
	h^{S} (X+Y) = h^{S} (X) + h^{S} (Y),
	$$
	and hence (b) holds.
	Next, we show (c). Using Axiom RF, we have $h^{S} (X-Y)\le 0$ if $X-Y\le 0$.
	Hence, by (b), we obtain (c). Moreover, (d) is implied by the equality
	$  A^{(S,0,\dots,0)}_1 = S $  from  \eqref{eq:trivial}. Finally, (e) follows from Axiom AF. 
	
	
Using Theorem \ref{thm:BT}, (a)-(e) imply that $h^S$ admits the representation
	$$h^S(X)=\E[X|S] \mbox{~for all $X\in \X$}.$$
For any $\XX \in \X^n$, let $S=S^\XX = \sum_{i = 1}^n X_i$. Using Axiom OA and the representation of $h^S$, we have 
	$$
	A_1^\XX = A_1^{(X_1, S -X_1, 0, \cdots, 0)} = h^{S }(X_1) = \E[X_1 | S ]= \E\left [X_1 | S^\XX \right].
	$$
	Similarly, we have $A_j^\XX = \E[X_j | S^\XX]$ for any $j = 2, \cdots, n$, which gives that $\mathbf{A}$ is CMRS.

	\noindent(ii) \underline{The case $\X=L^1_+$.} 
	
The gap between the proof in case  $\X=L^1$ 
and this case is that we need the following extension argument. 
	
	Let $\mathbf A$ be a risk sharing rule satisfying Axioms AF, RF, RA and OA. 
	Fix any $S \in \X = L^1_+$. Denote by $B_S=\{X\in L^1_+:  X\le S \}$, which is the set of random variables between $0$ and $S$.
	Define the mapping as in the proof of case (i),
	$$
	h^{S}: B_S \to L^1(\Omega, \sigma(S), \p), ~ X\mapsto A_1^{(X,S-X,0,\dots,0)}.
	$$
	It is clear that $h^{S}$ is well-defined on $B_S$ and satisfies additivity 
	$$
	h^S(X + Y) = h^S (X) + h^S(Y) \mbox{ for } X, Y, X+Y \in B_S,
	$$
	which can be  checked by the same argument as in the proof of Theorem \ref{th:full}.
	Define $C_S=\{\lambda X:\lambda \in \R_+ ,~X \in B_S \} $ which is the cone generated by $ B_S$, and $L_S=\{\lambda X:\lambda \in \R  ,~X \in B_S \} $  which is the linear space generated by $B_S$. 
	According to Lemma \ref{lem:extension} below,   $h^S$ can be uniquely extended on $C_S$ and $L_S$ and it is linear on $L_S$. This allows us to use the same arguments in (i) to get 
	$$
	h^S(X) = \E[X|S], ~~ X \in L_S,
	$$ 
and following the rest of the steps for the proof of Theorem \ref{th:full} yields that $\mathbf{A}$ is CMRS.
%
\end{proof}

The following lemma is used in the proof of Theorem \ref{th:full} above, where the sets $B_S$ and $C_S$ are defined.
\begin{lemma}\label{lem:extension}
	Fix $S\in L^1_+ $. 
	Any additive function $\phi: B_S\to L^1_+$
	has a unique additive extension on $C_S$ and a unique linear extension on $L_S$.
\end{lemma}

\begin{proof}[Proof]
	For $X\in C_S$, denote by $\gamma_X=\sup \{\gamma \in [0,1]: \gamma X\in B_S\}$.
	Note that  there exists $\lambda_X\in \R_+$ and $Y\in B_S$ such that $\lambda_X Y=X$, and hence $\gamma_X \ge 1/\lambda_X>0$.
	Moreover, we have $\gamma_X X \in B_S$ since $B_S$ is closed.  
	Define $\widehat \phi(X)= \phi(\gamma_X X)/\gamma_X$ for $X\in C_S$.
	It is clear that $\widehat \phi =\phi$ on $B_S$  because $\lambda_X=1$ for all $X\in B_S$.
	We next verify that $\widehat \phi$ is additive. 
	
	Take $m,k\in \N$  such that $m\le k$. 
	By additivity of $\phi$ on $B_S$, we have 
	$
	\phi (mX/k ) = m\phi ( X/k)
	$ for $X\in B_S$.
	By taking $m=1$, we get $\phi(X/k)=\phi(X)/k$,
	which in turn gives $
	\phi (mX/k ) = m\phi ( X)/k.
	$  
	Since $X$ is non-negative, positivity (monotonicity) of $\phi$ further gives $\phi(\lambda X)=\lambda \phi(X)$ for any real number $\lambda \in [0,1]$. 
	
	For any $X,Z\in C_S$ such that $Z\ge X$,  since $\gamma_Z \le \gamma_ X,$ we obtain, 
	by choosing $\lambda = \gamma_Z/\gamma_X$,  \begin{align}\label{eq:zdx}\phi(\gamma_Z X) = \phi(\lambda \gamma_X X) = \lambda\phi  (\gamma_X X )  = \frac{\gamma_Z}{\gamma_X} \phi  (\gamma_X X ).
	\end{align}
	Take any $X,Y\in C_S$ and  write $Z=X+Y$.
	Using \eqref{eq:zdx} and additivity of $\phi$ on $B_S$, 
	\begin{align*}
		\widehat \phi(X+Y) &=  \frac{1}{\gamma_{Z}}\phi(\gamma_{Z}( X+Y))\\
        & = \frac{1}{\gamma_{Z}}  \phi(\gamma_ZX) +\frac{1}{\gamma_{Z}}  \phi(\gamma_Z Y)\\
        &=
		\frac{1}{\gamma_{X}} \phi  (\gamma_X X ) +\frac{1}{\gamma_Y} \phi (\gamma_Y Y) 
		=\widehat \phi(X) + \widehat \phi(Y).   
	\end{align*}
	Therefore, $\widehat \phi$ is additive on $C_S$. 
	The extension is unique because any two additive and monotone functions agreeing on $B_S$ must agree on $C_S$. 
	The unique linear extension to $L_S$ follows from Theorem 1.10 of \cite{AB06}.
\end{proof}

\begin{proof}[Proof of Proposition \ref{prop:indep}]
	\begin{enumerate}[(i)]
	
		\item  The $Q$-CMRS rule $\mathbf{A}_{Q\text{-cm}}^\XX = \E^Q[\XX|S^\XX]$ satisfies Axioms RA, RF and OA with the same reasoning as the CMRS. Since $\E[ \mathbf{A}_{Q\text{-cm}}^\XX] =\E^Q[\XX]$, AF does not hold as long as $Q\ne \p$ with appropriate regularity conditions ($Q\sim\mathbb P$, and both $\mathrm{d} Q/\mathrm{d}  \mathbb{P}$ and $\mathrm{d}  \mathbb{P}/\mathrm{d} Q$ are essentially bounded).
		
		\item  For the mean-adjusted all-in-one risk sharing rule  
		$$
		\mathbf{A}_{\rm ma}^\XX =  \left(S^\XX-\E[S^\XX],0,\dots,0 \right)  +\E[\XX],
		$$ it is clear that Axioms  RA  and AF hold  by definition. 
		Axiom OA holds because the allocation to agent $i\in [n]$ is determined only by $(X_i,S)$.   Axiom RF does not hold because the allocation to agent $1$ is not a constant if $S^\XX$ is not a constant, regardless of whether $X_1$ is a constant, violating \eqref{eq:constancy}.
		
		\item  For the identity risk sharing rule $\mathbf{A}_{\rm id}^\XX = \XX$, it is clear that  
		Axioms AF, RF and OA hold. Axiom RA does not hold because $\mathbf{A}_{\rm id}^\XX$ is not necessarily a function of $S^\XX$.

		\item Consider a combination of $\mathbf{A}_{\rm all}$ and $\mathbf{A}_{\rm cm}$, defined by
	$\mathbf{A}^\XX = \mathbf{A}_{\rm all}^\XX =      (S^\XX , 0,\cdots,  0  )$ if  $\mathbf X$ is standard Gaussian, and $\mathbf{A}^\XX =\mathbf{A}_{\rm cm}^\XX = \E[\XX|S^\XX]$ otherwise. 
		Axioms  AF, RF and RA can be checked separately for   $\mathbf{A}_{\rm all}$ and $\mathbf{A}_{\rm cm}$, by noting that RF only needs to be checked for $\mathbf A_{\rm cm}$ since the standard Gaussian $\XX$  is not included in the statement of RF.
		
	 To verify that OA does not hold, it suffices to  consider $n=3$. Let $\XX=(X_1,X_2,X_3)$ follow  a standard Gaussian distribution. By definition, $A_1^\XX=S^\XX$. 
	 However, for $\mathbf Y=(X_1,X_2+X_3,0)$,
	 we have $A_1^{\mathbf Y}=\E[X_1|S^\XX]=S^\XX/3 \ne A_1^\XX$, thus violating OA.   
%
	\end{enumerate}
	
\end{proof}

\section{Proofs in Section \ref{sec:property}}

\begin{proof}[Proof of Theorem \ref{cor:UI}]
Since  $X\le_{\rm cx} Y$ implies   $X\le \sup X\le \sup Y$ and $\E[X]=\E[Y]$,  UI implies both RF and AF. Property CP follows from AF and RF as discussed in \eqref{eq:constancy}.
	The proof  of the last statement 
	follows directly from the above fact, Theorem \ref{th:full}, and the fact that CMRS satisfies Property UI.
\end{proof}

\begin{proof}[Proof of Proposition \ref{prop:BT}]
We only show the case that $\X=L^1$, as the case $\X=L^1_+$ is analogous. 
	Let $\mathbf A: \X^n \to \X^n$ be a risk sharing rule satisfying Axioms AF, RF and OA.  
	Fix $S\in \X$. For any $X\in L^1(\Omega,\sigma(S),\p)$,   $A_1^{(X,S-X,0,\dots,0)}$ is $\sigma(S)$-measurable,  because it is $\sigma(X,S-X)$-measurable by assumption and $\sigma(S)=\sigma(X,S-X)$.
	Define the mapping
	$$
	h^{S}: L^1(\Omega, \sigma (S ), \p) \to L^1(\Omega, \sigma(S), \p), ~ X \mapsto A_1^{(X,S-X,0,\dots,0)},
	$$
	where we use the fact that $h^{S}(X)$ is $\sigma(S)$-measurable for $X\in L^1(\Omega, \sigma (S ), \p)$. 
	The arguments in the proof of Theorem \ref{th:full} yield that $h^{S }$ satisfies the conditions in Theorem \ref{thm:BT}. 
	By Theorem \ref{thm:BT}, $h^{S }$ is the identity on $L^1 (\Omega, \sigma (S    ), \p )$. Using this and   OA, 
	for any $\XX\in \mathbb A_n(S)$ and $X_1\in L^1(\Omega,\sigma(S),\p)$,  we have
	$$
	A_1^{\XX} = A_1^{(X_1, S-X_1, 0, \cdots, 0)} = h^{S}(X_1) = X_1.
	$$ 
	The other case of $A_j^\XX$ for $j\in [n]$ are similar. 
\end{proof}


\begin{proof}[Proof of Proposition \ref{prop:SM}]
By Proposition \ref{prop:OA}, OA implies that $A_i^\XX$ is determined by $(X_i,S^\XX)$ and $i \in [n]$. It suffices to show that $i\in [n]$ is also not relevant.  
Using ZP and OA, we have, for the pair $(1,3)$ and any $X,S\in \X$,
$$
A_1^{(X,S-X,0,\dots,0)}= S-A_2^{(X ,S-X,0,\dots,0)} = S-A_2^{(0,S-X, X,0,\dots,0)} =  A_3^{(0,S-X,X,0,\dots,0)}.$$
The other pairs $(i,j)$  are similar. Therefore, $A_i^\XX$ is determined by $(X_i,S^\XX)$ regardless of $i \in [n]$, showing that SM holds. 
\end{proof}

\section{Proofs and an Example in Section \ref{sec:equ}}
\label{app:4}
\begin{proof}[Proof of Theorem \ref{prop:IR}]
First,
define decreasing convex order: $X\le_{\rm icx} Y$ if $\E[v(X)]\le \E[v(X)]$ for all decreasing convex functions $v$. 
Since decreasing convex order is characterized by risk averse agents,  [E1] is can be restated as Property UI with convex order replaced by decreasing convex order.
If $\mathbb A$ satisfies [E1],
then, by $\sum_{i=1}^n \E[A_i^\XX]=\E[S^XX]=\sum_{i=1}^n  \E[X_i]$
and $\E[A_i^\XX]\ge \E[X_i]$ for each $i\in [n]$, we obtain $\E[A_i^\XX]=\E[X_i]$.
Therefore,  [E1] is equivalent to UI.

Second, [E3] is equivalent to Axiom OA, because   indifference under all agents is equivalent to equality between random variables.  
It remains
to show that   [E2] is equivalent to Axiom RA.  
Let $\mathbf A$ be a risk sharing rule satisfying Axiom RA. For  $\XX\in \X^n$ and a $\sigma$-field $\mathcal G$, suppose that an allocation of $S^\XX$, $\mathbf Y=(Y_1,\dots,Y_n)=\E[ \mathbf A^\XX |\mathcal G]$, strictly improves $\mathbf A^\XX$.
Since $$S^\XX= \sum_{i=1}^n Y_i =\sum_{i=1}^n \E[ A^\XX_i  |\mathcal G]=\E[S^\XX|\mathcal G],$$
we have that $S^\XX$ is $\mathcal G$-measurable. Using Axiom RA, this implies $\E[ A^\XX_i |\mathcal G] =A^\XX_i$, a contradiction. Thus, [E2] holds.

Next, we show the converse direction. 
For any allocation $\mathbf A$, if $\mathbf A^\XX$ is not $\sigma(S^\XX)$-measurable, then 
$\E[\mathbf A^\XX|S^\XX]$ strictly improves $\mathbf A^\XX$ in convex order, and hence strictly risk averse agents strictly prefer components of $\E[\mathbf A^\XX|S^\XX]$ over those of $\mathbf A^\XX$.
Therefore, [E2] implies Axiom RA.  \end{proof}

\begin{proof}[Proof of Proposition \ref{coro:conflict}]
Fix  a non-constant $S\in \X$, and write $h^{S}:\X\to \X, ~ X\mapsto A_1^{(X,S-X,0,\dots,0)}$.
By ZP, we have $h^S(S)=S$. 
Using additivity  \eqref{eq:rw1} guaranteed by OA in the proof of Theorem \ref{th:full}, we have 
$h^S( 2S)=2 h^S(S) =2 S$. 
This and additivity give $h^S(-S)=-S$, and 
therefore,   $\mathbf A^{(-S,2S,0,\dots,0)}=(-S,2S,0,\dots,0)$ is not comonotonic.  
\end{proof}

\begin{example}[Comonotonicity is optimal for risk-averse dual utility agents]
\label{ex:Yaari}
Recall that a risk-averse dual utility $\mathcal U:\X\to \R$ is comonotonic-additive and superadditive; see \cite{Y87} and \citet[Theorems 1 and 3]{WWW20}. For any comonotonic risk vector $\XX$ and any $\mathbf Y=(Y_1,\dots,Y_n)\in \mathbb A_n(S^\XX)$,  
	$$
	\sum_{i=1}^n \mathcal U(- X_i ) = \mathcal U\left( -\sum_{i=1}^n X_i  \right)   = \mathcal U\left(- \sum_{i=1}^n Y_i  \right)  \ge \sum_{i=1}^n \mathcal U(-Y_i  ),
	$$
	where the negative sign reflects that our risks are losses and $\mathcal U$ is applied to wealth. 
	Therefore, $\mathbf Y$ does not dominate $\mathbf X$, showing that $ \XX$ is Pareto optimal. 
	\end{example}

\section{Proofs and Further Examples in Section \ref{sec:equilibrium}}

\begin{proof}[Proof of Theorem \ref{prop:general-cmrs}.]

		The hypothesis states that there exists a positive $\sigma(S^{\mathbf X})$-measurable random variable $M(S^{\mathbf X})$ such that $\lambda_i U_i'(-\mathbb E[X_i| S^{\mathbf X}]) = M(S^{\mathbf X})$ for all $i$. This is precisely the first-order condition \eqref{eq:first-order-equilibrium} characterizing Pareto optimal allocations. Since utilities are strictly concave, this allocation is the unique Pareto optimum associated with the weights $\lambda_i$. 
		In a complete Arrow-Debreu market, a Pareto optimal allocation can be supported as an equilibrium via a positive pricing kernel (\citealt{B62}), 
        $Y_i^* = \E^\Psi[X_i | S^{\mathbf{X}}]$. 
        Furthermore, since $\Psi$ is a function of $S^{\mathbf X}$ only, it is $\sigma(S^{\mathbf X})$-measurable. Thus, for all $i$, we have
        $$
        {
        \E^\Psi[X_i | S^{\mathbf{X}}] = \frac{\E [\Psi X_i | S^{\mathbf{X}}]}{\E [\Psi | S^{\mathbf{X}}]} = \frac{\Psi \cdot \E[X_i | S^{\mathbf{X}}]}{\Psi} =
        \E[X_i | S^{\mathbf{X}}] := Y_i^*.}
        $$
		Hence,  
		$
		\mathbb E[\Psi Y_i^*] = \mathbb E[\Psi \mathbb E[X_i| S^{\mathbf X}]] = \mathbb E[\mathbb E[\Psi X_i| S^{\mathbf X}] ] = \mathbb E[\Psi X_i].
		$ 
		Thus, the budget constraints hold: $\mathbb{E}[\Psi Y_i^*] = \mathbb{E}[\Psi X_i]$ for all $i$.  
		Since the allocation maximizes individual utilities and satisfies the budget constraints, it is the unique  equilibrium allocation.
        The last statement follows from the fact that  when components of 
        the random vector ${(\E[X_1 | S^{\mathbf{X}}], \dots, \E[X_n | S^{\mathbf{X}}])}$ are {strictly} 
        increasing functions of each other, we can find $U_1,\dots,U_n$ such that 
        \eqref{eq:first-order-equilibrium} holds.
\end{proof}

Beyond the quota-share form $Y_{i}^{\ast }=a_{i}S^{\mathbf X}+b_{i}$, the third example concerns endowments with mean-zero noise,  given by $X_i = f_i(S^{\mathbf X}) + \varepsilon_i$ for some strictly increasing functions $f_1,\dots,f_n$, where the noise terms satisfy $\mathbb E[\varepsilon_i | S^{\mathbf X}] = 0$; this implies $\sum_{i=1}^n \varepsilon_i = 0$ because $S^{\mathbf X}=\sum_{i=1}^n f_i(S^{\mathbf X})+\sum_{i=1}^n \varepsilon_i $. Then,  
$
\mathbb E[X_i | S^{\mathbf X}] = f_i(S^{\mathbf X}),
$
which satisfies the   condition in the last statement of Theorem \ref{prop:general-cmrs}. Hence, the equilibrium allocation   coincides with the CMRS rule for some specific utility functions in the condition \eqref{eq:first-order-equilibrium}.

    Section \ref{sec:equilibrium} establishes that under specific alignment conditions between preferences and endowments, the CMRS rule emerges as the unique competitive equilibrium, often taking a linear quota-share form. In this appendix, we explore the boundaries of these results by examining the impact of market incompleteness. Using a mean-variance framework (quadratic utility), we demonstrate that the equilibrium risk sharing rule is still linear in complete markets, but this no longer holds in incomplete markets where agents trade only a limited set of securities. Further, we reveal a counter-intuitive phenomenon: in incomplete markets, the introduction of a new financial security can inadvertently amplify individual risk exposures rather than diversify them. 

We consider a two-agent economy ($n=2$) where each agent $i \in \{1, 2\}$ is initially endowed with an independent normally distributed  $X_i \sim \mathrm{N}(0, \sigma_i^2)$. The agents' behavior is governed by the mean-variance criterion. Agent $i$ chooses positions $\phi$ in securities to maximize:
$$
\max_{\phi} \left\{ \mathbb{E}[Y_i] - \frac{\theta_i}{2} \text{Var}(Y_i) \right\},
$$
where $\theta_i > 0$ represents Agent $i$'s absolute risk aversion coefficient and the final risk position $Y_i$ is the sum of the initial endowment and the payoff from the traded securities, adjusted for the market prices. The market prices and the equilibrium positions are determined endogenously by the market clearing condition, requiring that the net supply of each security is zero. This setup allows us to explicitly solve for the equilibrium allocation and analyze its structural properties. 

    
	First, we start from a complete market by introducing two securities with random 
	payoffs $-X_{1}$ and $-X_{2}$ and constant prices $p_{1}$ and $p_2$. Then the new risk after taking
	positions $\phi _{i,1}$ and $\phi _{i,2}$ in each security is given by 
	$$
	Y_{1} = X_{1}+\phi _{1,1}(-X_{1}-p_{1})+\phi _{1,2}(-X_{2}-p_{2}),\quad
	Y_{2}=X_{2}+\phi _{2,1}(-X_{1}-p_{1})+\phi _{2,2}(-X_{2}-p_{2}).
	$$%
	For Agent 1, the
	objective is
	$\max_{\phi_{1,1}, \phi _{1,2}} \left\{\E[Y_{1}]-\frac{\theta _{1}}{2}%
	\var(Y_{1})\right\}$ or
	$$
	\max_{\phi _{1,1},\phi _{1,2}} \left\{-\phi _{1,1}p_{1}-\phi
	_{1,2}p_{2}-\frac{\theta _{1}}{2}\left[ (1-\phi _{1,1})^{2}\sigma
	_{1}^{2}+\phi _{1,2}^{2}\sigma _{2}^{2}\right]\right\},
	$$
	leading to the optimal positions for Agent 1 as $\phi _{1,1}^{\ast }=1-\frac{%
		p_{1}}{\theta _{1}\sigma _{1}^{2}}$ and $\phi _{1,2}^{\ast }=-\frac{p_{2}}{%
		\theta _{1}\sigma_{2}^{2}}.$ Similarly, the optimal positions for Agent 2
	are $\phi_{2,1}^{\ast }=-\frac{p_{1}}{\theta _{2}\sigma _{1}^{2}}$ and $%
	\phi _{2,2}^{\ast }=1-\frac{p_{2}}{\theta _{2}\sigma _{2}^{2}}$. The market
	clearance condition requires that in equilibrium $0=\phi_{1,1}^{\ast }+\phi
	_{2,1}^{\ast}$ and $0=\phi _{1,2}^{\ast }+\phi _{2,2}^{\ast}$, implying
	that the equilibrium prices are 
	$
	p_{1}=\sigma _{1}^{2}\frac{\theta _{1}\theta _{2}}{\theta _{1}+\theta _{2}}$
	and 
	$p_{2}=\sigma _{2}^{2}\frac{\theta _{1}\theta _{2}}{\theta _{1}+\theta
		_{2}},
	$
	with the resulting equilibrium positions being
	$$
	\phi_{1,1}^{\ast }=\frac{\theta _{1}}{\theta _{1}+\theta_{2}}, \;\;
    \phi
	_{1,2}^{\ast }=-\frac{\theta _{2}}{\theta _{1}+\theta _{2}}, \;\;
	\phi
	_{2,1}^{\ast }=-\frac{\theta _{1}}{\theta _{1}+\theta _{2}}, \;\;
    \phi
	_{2,2}^{\ast }=\frac{\theta _{2}}{\theta _{1}+\theta _{2}}.
	$$ 
	Thus, 
	$$
    \left\{
    \begin{aligned}
	&Y_{1}^{\ast }= \underbrace{\frac{\theta _{2}}{\theta _{1}+\theta _{2}}}_{\text{risk tolerance share}}\left(
	X_{1}+X_{2}\right) + \underbrace{(\sigma _{2}^{2}\theta _{2}-\sigma _{1}^{2}\theta _{1})%
	\frac{\theta _{1}\theta _{2}}{\left( \theta _{1}+\theta _{2}\right) ^{2}}}_{\text{side payment}},\\
    &Y_{2}^{\ast }=\underbrace{\frac{\theta _{1}}{\theta _{1}+\theta _{2}}}_{\text{risk tolerance share}}\left(
	X_{1}+X_{2}\right) +\underbrace{(\sigma _{1}^{2}\theta _{1}-\sigma _{2}^{2}\theta _{2})%
	\frac{\theta _{1}\theta _{2}}{\left( \theta _{1}+\theta _{2}\right) ^{2}}}_{\text{side payment}},
    \end{aligned}
	\right.
    $$%
	satisfying the linear quota-share form. 
	
	Next, consider an incomplete market, in which only one (aggregated) security is
	available to trade, whose payoff is $-X_{1}-X_{2}$. Let $\phi _{i}$ be the quantity held in this security. 
    Then the final payoff after introducing
	the security is 
	$
	Y_{i}=X_{i}+\phi _{i}(-X_{1}-X_{2}-p),
	$
	where $p$ is the price of the security. For Agent 1, the objective is
	$
	\max_{\phi_{1}} \left\{\E[Y_{1}]-\frac{\theta _{1}}{2}\var(Y_{1}) \right\}$ 
	or 
	$$
    \max_{\phi
		_{1}}\left\{ -\phi _{1}p-\frac{\theta _{1}}{2}\left[ (1-\phi _{1})^{2}\sigma
	_{1}^{2}+\phi _{1}^{2}\sigma _{2}^{2}\right] \right\}.
	$$
	Thus, we have $\phi _{1}^{\ast }=-\frac{p-\theta _{1}\sigma _{1}^{2}}{\theta
		_{1}\left( \sigma _{1}^{2}+\sigma _{2}^{2}\right) }$. Similarly, $\phi
	_{2}^{\ast }=-\frac{p-\theta _{2}\sigma _{2}^{2}}{\theta _{2}\left( \sigma
		_{1}^{2}+\sigma _{2}^{2}\right) }$. Since in equilibrium $\phi _{1}^{\ast
	}+\phi _{2}^{\ast }=0,$ we have the equilibrium price is
	$
	p=
	\frac{\theta _{1}\theta _{2}}{\theta _{1}+\theta
		_{2}}(\sigma _{1}^{2}+\sigma _{2}^{2}),
	$
	with the resulting trading positions being
	$$
	\phi _{1}^{\ast }=\frac{\sigma _{1}^{2}}{\sigma _{1}^{2}+\sigma _{2}^{2}}-%
	\frac{\theta _{2}}{\theta _{1}+\theta _{2}}, \;\;
	\phi _{2}^{\ast }=\frac{%
		\sigma _{2}^{2}}{\sigma _{1}^{2}+\sigma _{2}^{2}}-\frac{\theta _{1}}{\theta
		_{1}+\theta _{2}}.
	$$
	Thus,
	$$
    \left\{
    \begin{aligned}
	&Y_{1}^{\ast }=X_{1}-\left( \frac{\sigma _{1}^{2}}{\sigma _{1}^{2}+\sigma
		_{2}^{2}}-\frac{\theta _{2}}{\theta _{1}+\theta _{2}}\right) \left(
	X_{1}+X_{2}+	\frac{\theta _{1}\theta _{2}}{\theta _{1}+\theta
		_{2}}(\sigma _{1}^{2}+\sigma _{2}^{2})\right),\\
	&Y_{2}^{\ast }=X_{2}-\left( \frac{\sigma _{2}^{2}}{\sigma _{1}^{2}+\sigma
		_{2}^{2}}-\frac{\theta _{1}}{\theta _{1}+\theta _{2}}\right) \left(
	X_{1}+X_{2}+ 	\frac{\theta _{1}\theta _{2}}{\theta _{1}+\theta
		_{2}}(\sigma _{1}^{2}+\sigma _{2}^{2})
	\right).
    \end{aligned}
    \right.
	$$%
	Note that the equilibrium risk sharing rule does not depend solely on the total sum $X_{1}+X_{2}$
	(Axiom RA fails to hold) and does not satisfy the linear form. 
	
	Additionally, an interesting insight is that, in the case of limited securities available
	to trade, the exposure to $X_{i}$ after the equilibrium risk sharing
	may be higher than that without risk sharing. 
    Indeed, since $\phi^*_1 + \phi^*_2 = 0$, exactly one of the two coefficients can be negative. Without loss of generality, if $\frac{\sigma _{1}^{2}}{\sigma _{1}^{2}+\sigma _{2}^{2}}-\frac{%
		\theta _{2}}{\theta _{1}+\theta _{2}}<0$, then Agent 1's exposure to $X_1$ is amplified beyond its pre-trade level.
    In other words, the financial innovation of introducing a security may amplify
	the speculation. Our result complements the results in \cite{Simsek2013}, in
	which the main reason for amplification is heterogeneous beliefs. Here the
	agents may have the same beliefs about $X_{1}$ and $X_{2}$, but the amplification comes from
	market incompleteness. To ensure clarity, key insights are presented under the simple setup of $n=2$. These insights extend to high-dimensional settings, where deriving the optimality condition via matrices becomes considerably more tedious.

\section{Proofs in Section \ref{sec:generalized}}
\begin{proof}[Proof of Theorem \ref{th:full3}]
In Section \ref{sec:axiom}, we   
used properties of the conditional expectation to check Axioms AF, RF  and OA for CMRS, and the same properties hold for the generalized CMRS.
Properties IA and IB are straightforward from   basic properties of $ \E  [\XX | \mathcal{G}^\XX ]$.
Therefore, the ``if" statement holds.

 We proceed to prove the ``only if" statement. 
 We will only show the case $\X=L^1$. The case $\X=L^1_+$ is similar as we see in the proof of Theorem \ref{th:full}.
 Let $\widehat{\mathbf A}$  be a generalized risk sharing rule satisfying Axioms AF, RF and OA and Properties IA and IB.
	Fix any $S \in \X$. 
	Define the mapping
	$$
	h^{S|\mathcal G}:\X\to L^1(\Omega, \mathcal{G}^{S}, \p), ~ X\mapsto \widehat A_1^{(X,S-X,0,\dots,0)|\mathcal G},
	$$
	where $\mathcal G^S=\sigma(S,\mathcal G)$.
	Note that IA guarantees that $h^{S|\mathcal G}$ takes values  in  $L^1(\Omega, \mathcal{G}^S, \p)$. 
	We will verify that $h^{S|\mathcal G}$  satisfies the following properties on $\X$:
	\begin{enumerate}[(a)]
		\item constant preserving: $h^{S|\mathcal G}(t)=t$ for all $t\in \R$; 
		\item additivity: $h^{S|\mathcal G}  (X+Y) = h^{S|\mathcal G}  (X) + h^{S|\mathcal G} (Y) $ for $X,Y\in \X$; 
		\item monotonicity: $h^{S|\mathcal G} (Y) \ge h^{S|\mathcal G} (X)$ if $Y\ge X$;
		\item $h^{S|\mathcal G} (X) = X$ for any $X $ that is $\mathcal{G}^S$-measurable;
		\item $\E[h^{S|\mathcal G}(X)] = \E[X]$ for $X \in \X$. 
	\end{enumerate}
	The properties (a)-(c) and (e) can be  shown analogously  to the proof of Theorem \ref{th:full}. The property (d) is implied by Property IB. 
	Note that  (d) corresponds to  \eqref{eq:SBT} in the proof of  Theorem \ref{thm:BT}.
	Using the same argument there, we obtain that (a)-(e) imply that $h^{S|\mathcal G}$ admits the representation
	$$h^{S|\mathcal G}(X)=\E\left [X|\mathcal{G}^S\right] \mbox{~for all $X\in \X$}.$$
For any $\XX \in \X^n$, let $S=S^\XX = \sum_{i = 1}^n X_i$. Using Axiom OA and the representation of $h^{S|\mathcal G}$, we have 
	$$
	\widehat A_1^\XX = \widehat A_1^{(X_1, S -X_1, 0, \cdots, 0)} = h^{S|\mathcal G }(X_1) = \E\left[X_1 | \mathcal{G}^\XX \right].
	$$
	Similarly, we have $A_j^\XX = \E[X_j | \mathcal{G}^\XX]$ for any $j = 2, \cdots, n$, which completes the proof.
\end{proof}

\section{Proofs in Section \ref{sec:app}}
\begin{proof}[Proof of Proposition \ref{prop:bitcoin}]
We have seen that CMRS satisfies the four axioms for mappings on general spaces, and hence also on $\mathcal B_n$, and it clearly satisfies the definition of a reward sharing rule. Below we show that the four axioms are   sufficient for CMRS. 
	Fix $D \subseteq \Omega$ independent of $P$ and denote by  
	$$
	I_D=\{C \subseteq D: \mbox{$C$ is independent of $P$}\} \mbox{~~and~~}M_{D} = \{ P\id_C: C \in I_D\}.
	$$   
	Define the mapping as in the proof of Theorem \ref{th:full},
	$$
	h^{D}: M_{D} \to L^1(\Omega, \sigma(P \id_D), \p), ~ P \id_C \mapsto A_1^{(P \id_C, P(\id_D-\id_C),0,\dots,0)}.
	$$
	We can check that $h^{D}$  satisfies additivity
	$$
	h^{D}(X + Y) = h^{D} (X) + h^{D}(Y) \mbox{ for } X, Y  \in M_{D} \mbox{ with } X+Y\in M_D,
	$$
	and monotonicity 	$$
	h^{D}(X )  \le h^{D}(Y) \mbox{ for } X, Y  \in M_{D} \mbox{ with } X\le Y;
	$$
	these statements can be shown by   arguments using AF,  RF and OA as in the proof of Theorem \ref{th:full}.
	For $m\in \N$, take $C_1,\dots,C_m \in I_D$ such that $\p(C_1)=\dots=\p(C_m) $ and $\bigcup_{j=1}^m C_j=D$.
	Using $h^{D}(P\id_D) = P\id_D$ (guaranteed by AF and RF) and $h^D(P\id_{C_1}) = h^D(P\id_{C_j})$ for $j\in [m]$ (by the definition of a reward sharing rule), we get from additivity of $h^{D}$ that $$h^D(P\id_{C_1}) = \frac{P\id_D}{m} = P\id_D \frac{\p(C_1)}{\p(D)}.$$
	Since $C_1$ is arbitrary, we get that, for any $ C\in I_D$ with $\p(C)= \p(D) /m$,  
	\begin{align}\label{eq:bitcoin-1} h^D(P\id_{C}) = P\id_D \frac{\p(C)}{\p(D)} = P \id_D\E[ \id_{C} |P\id _D]= \E[P\id_{C} |P\id _D].\end{align} 
Using additivity again, we know that \eqref{eq:bitcoin-1} holds for any $j\in [m]$ and $C\in I_D$ with $\p(C)=j\p(D) /m$,  and finally, by  monotonicity of $h^D$, we get \eqref{eq:bitcoin-1} for all $C\in I_D$. 
Following the rest of the steps for the proof of Theorem \ref{th:full} yields that $\mathbf{A}$ is CMRS on $\mathcal B_n$.  
	\end{proof}

    \end{document}